\documentclass[12pt,oneside,reqno]{amsart}
\usepackage{graphicx}
\usepackage{caption,subcaption}
\usepackage{dcolumn}
\usepackage{amsmath,amsthm,amssymb}
\usepackage{url}
\usepackage{latexsym, graphicx,cite}

\usepackage[colorlinks=true, linkcolor=blue, bookmarks=true]{hyperref}

\usepackage{euscript}

\newcommand{\beq}{\begin{equation}}
\newcommand{\eeq}{\end{equation}}
\newcommand{\ber}{\begin{array}}
\newcommand{\eer}{\end{array}}
\newcommand{\del}{\partial}

\newcommand{\ena}{\end{eqnarray}}
\newcommand{\beqa}{\begin{eqnarray}}
\newcommand{\eeqa}{\end{eqnarray}}
\newcommand{\bea}{\begin{eqnarray}}
\newcommand{\eea}{\end{eqnarray}}

\makeatletter\@addtoreset {equation}{section}\makeatother

\newtheorem{theorem}{Theorem}

\newtheorem{lemma}{Lemma}
\theoremstyle{remark}
\newtheorem{rem}{Remark}

\renewcommand{\Im}{\operatorname{Im}}

\newcommand{\minPlusOne}[1]{
  \min\!\left(#1\right) + 1
}

\usepackage{amsthm}
\usepackage{latexsym}
\usepackage{amssymb}
\usepackage{amsmath}
\usepackage{mathrsfs}

\begin{document}

\title[Stationary states in the conformal flow]
{Stationary states of the cubic conformal flow on $\mathbb{S}^3$}

\author{Piotr Bizo\'n}
\address{Institute of Physics, Jagiellonian
University, Krak\'ow, Poland}
\email{bizon@th.if.uj.edu.pl}

\author{Dominika Hunik-Kostyra}
\address{Institute of Physics, Jagiellonian
University, Krak\'ow, Poland}
\email{dominika.hunik@uj.edu.pl}

\author{Dmitry Pelinovsky}
\address{Department of Mathematics, McMaster University, Hamilton, Ontario  L8S 4K1, Canada}
\email{dmpeli@math.mcmaster.ca}

\thanks{This research was supported by the Polish National Science Centre grant no. 2017/26/A/ST2/00530.}

\date{\today}

\begin{abstract}
We consider the resonant system of amplitude equations for the conformally invariant cubic wave equation
on the three-sphere. Using the local bifurcation theory, we characterize all stationary states
that bifurcate from the first two eigenmodes. Thanks to the variational formulation of the resonant system and
energy conservation, we also determine variational characterization and stability of the bifurcating states.
For the lowest eigenmode, we obtain two orbitally stable families of the bifurcating stationary states:
one is a constrained maximizer of energy and the other one is a constrained minimizer of the energy,
where the constraints are due to other conserved quantities of the resonant system.
For the second eigenmode, we obtain two constrained minimizers of the energy, which are also
orbitally stable in the time evolution. All other bifurcating states are saddle points of energy under these constraints
and their stability in the time evolution is unknown. 
\end{abstract}

\maketitle
\tableofcontents

\section{Introduction}

The conformally invariant cubic wave equation on $\mathbb{S}^3$ (the unit three-dimensional sphere)
is a toy model for studying the dynamics of resonant interactions between nonlinear waves
on a compact manifold. The long-time behavior of small solutions of this equation
is well approximated by solutions of an infinite dimensional time-averaged Hamiltonian system,
called the \emph{cubic conformal flow}, that was introduced and studied in \cite{bal}. In terms of complex amplitudes
$\{ \alpha_n(t) \}_{n \in \mathbb{N}}$, this system  takes the form
\beq
i(n+1) \frac{d \alpha_n}{d t} =
\sum\limits_{j=0}^\infty \sum_{k=0}^{n+j} S_{njk,n+j-k}\,\bar\alpha_j \alpha_k \alpha_{n+j-k}\,,
\label{flow}
\eeq
where $S_{njk,n+j-k} = \minPlusOne{n,j,k,n+j-k}$ are the interaction coefficients. The
cubic conformal flow (\ref{flow}) is the Hamiltonian system with the symplectic form
$\sum_{n=0}^{\infty} 2i (n+1) d \bar{\alpha}_n \wedge (-d \alpha_n)$ and
the conserved energy function
\begin{equation}
H(\alpha) = \sum\limits_{n=0}^{\infty}\sum\limits_{j=0}^{\infty}
\sum\limits_{k=0}^{n+j} S_{njk,n+j-k}\,\bar\alpha_n \bar \alpha_j \alpha_k \alpha_{n+j-k}.
\label{Hconf}
\end{equation}

The attention of \cite{bal} has been focused on understanding the patterns of energy transfer between the modes.
In particular, a three-dimensional invariant manifold was found on which the dynamics is Liouville-integrable
with exactly periodic energy flows. Of special interest are the stationary states  for which no transfer
of energy occurs. A wealth of explicit stationary states have been found in \cite{bal}, however a complete
classification of stationary states was deemed as an open problem.

The purpose of this paper is to study existence and stability of all stationary states
that bifurcate from the first two eigenmodes of the cubic conformal flow. Thanks to the Hamiltonian
formulation, we are able to give the variational characterization of the bifurcating families.
Among all the families, we identify several particular stationary states
which are orbitally stable in the time evolution: one is a constrained maximizer of energy
and the other ones are local constrained minimizers of the energy, where the constraints are induced by
other conserved quantities of the cubic conformal flow.

The constrained maximizer of energy can be normalized to the form
\beq
\label{ground-intro}
\alpha_n(t) = (1-p^2) p^n e^{-i t},
\eeq
where $p \in (0,1)$ is a parameter. This solution was labeled as the {\em ground state} in our previous work \cite{BHP18},
where we proved orbital stability of the ground state in spite of its degeneracy with respect to
parameter $p \in (0,1)$.

Two constrained minimizers of energy are given by the exact solutions:
\beq
\label{min-intro}
\alpha_n(t) = c \left[ (1-p^2) n - \frac{1}{2} (1 + 5 p^2 \pm \sqrt{1-14p^2 + p^4}) \right] p^n e^{-i \lambda t + i n \omega t},
\eeq
where $p \in (0,2-\sqrt{3})$ and $c \in (0,\infty)$ are parameters, whereas $(\lambda,\omega)$
are expressed by
\begin{equation}\label{omega-lambda-intro}
\lambda  =\frac{c^2}{6}\, \left(\frac{3-4p^2}{1-p^2}
\pm \frac{(3+4p^2) \sqrt{1-14p^2 + p^4}}{(1-p^2)^2}\right), \quad
\omega  = \frac{c^2}{12} \, \frac{1+p^2 \pm \sqrt{1-14p^2 + p^4}}{1-p^2}.
\end{equation}
The cutoff in the interval $(0,2-\sqrt{3})$ for $p$ ensures that $1-14p^2 + p^4 > 0$.
Since the constrained minimizers of energy are nondegenerate with respect to $(\lambda,\omega)$,
their orbital stability follows from the general stability theory \cite{GSS90}.

By using the local bifurcation methods of this paper, we are able to prove that
{\em the stationary states (\ref{ground-intro}) and (\ref{min-intro}) are respectively
maximizer and two minimizers of energy constrained by two other conserved quantities in the limit of small $p$.}
For the ground state (\ref{ground-intro}), we know from our previous work \cite{BHP18}
that it remains a global constrained maximizer of energy for any $p \in (0,1)$.
For the stationary states (\ref{min-intro}) and (\ref{omega-lambda-intro}), we have
checked numerically that they remain local constrained minimizers of energy
for any $p \in (0,2-\sqrt{3})$ and $c > 0$, however, we do not know if
any of them is a global constrained minimizer of energy.

In addition, we have proven {\em existence of another constrained minimizer of energy bifurcating from the second eigenmode}.
The new minimizer is nondegenerate with respect to $(\lambda,\omega)$, hence again its
orbital stability follows from the general stability theory in \cite{GSS90}.
However, we show numerically that this stationary state remains a local constrained minimizer only near the bifurcation
point and becomes a saddle point of energy far from the bifurcation point.

Bifurcation analysis of this paper for the first two eigenmodes suggests
existence of other constrained minimizers of energy bifurcating from other
eigenmodes. This poses an open problem of characterizing a global constrained
minimizer of energy for the cubic conformal flow (\ref{flow}). Another
open problem is to understand orbital stability of the saddle points
of energy, in particular, to investigate if other conserved quantities might
contribute to stabilization of saddle points of energy.

The cubic conformal flow \eqref{flow} shares many properties with a cubic resonant system for
the Gross-Pitaevskii equation in two dimensions \cite{GHT, BBCE}. In particular,
the stationary states for the lowest Landau level invariant subspace of the cubic resonant system
have been thoroughly studied in \cite{GGT} by using the bifurcation theory from a simple eigenvalue \cite{CR71}.
In comparison with Section 6.4 in \cite{GGT}, where certain symmetries were imposed to reduce
multiplicity of eigenvalues, we develop normal form theory for bifurcations
of all distinct families of stationary states from the double eigenvalue without imposing
any a priori symmetries.

Another case of a completely integrable resonant system with a wealth of stationary states
is the cubic Szeg\H{o} equation \cite{GG,GG2012}. Classification of stationary states and their stability
has been performed for the cubic and quadratic Szeg\H{o} equations in \cite{Poc} and \cite{Thir} respectively.

One more example of a complete classification of all travelling waves of finite energy
for a non-integrable case of the energy-critical half-wave map onto $\mathbb{S}^2$ is given
in \cite{LS}, where the spectrum of linearization at the travelling waves is studied
by using Jacobi operators and conformal transformations. The half-wave map was found to
be another integrable system with the Lax pair formulation \cite{GL}.
It is unclear in the present time if the cubic conformal flow on $\mathbb{S}^3$ is also
an integrable system with the Lax pair formulation. \\

{\bf Organization of the paper.} Symmetry, conserved quantities,
and some particular stationary states for the cubic conformal flow (\ref{flow})
are reviewed in Section \ref{sec-prel} based on the previous works \cite{bal,BHP18}.
Local bifurcation results including the normal form computations for the lowest eigenmode
are contained in Section 3. Variational characterization of the bifurcating families
from the lowest eigenmode including the proof of extremal
properties for the stationary states (\ref{ground-intro}) and (\ref{min-intro})
is given in Section 4. Similar bifurcation results and variational characterization
of the bifurcating states from the second eigenmode are obtained in Sections 5 and 6
respectively. Numerical results confirming local minimizing properties
of the stationary states (\ref{min-intro}) for all admissible values of $p$
are reported in Section 7.\\

{\bf Notations.} We denote the set of nonnegative integers by $\mathbb{N}$ and
the set of positive integers by $\mathbb{N}_+$.
A sequence $(\alpha_n)_{n\in \mathbb{N}}$ is denoted for short by $\alpha$.
The space of square-summable sequences on $\mathbb{N}$ is denoted by~$\ell^2(\mathbb{N})$.
It is equipped with the inner product $\langle \cdot, \cdot \rangle$ and the induced norm
$\| \cdot \|$. The weighted space $\ell^{2,2}(\mathbb{N})$ denotes the space
of squared integrable sequences with the weight $(1+n)^2$. We write $X \lesssim Y$
to state $X\leq C Y$ for some universal
(i.e., independent of other parameters) constant $C>0$.
Terms of the Taylor series in $(\epsilon,\mu)$
of the order $\mathcal{O}(\epsilon^p,\mu^p)$ are denoted by $\mathcal{O}(p)$.

\section{Preliminaries}
\label{sec-prel}

Here we recall from \cite{bal,BHP18} some relevant properties
of the cubic conformal flow (\ref{flow}) and its stationary states.
The cubic conformal flow (\ref{flow}) enjoys the following three one-parameter groups of symmetries:
 \begin{eqnarray}
 \mbox{Scaling:} \,&& \alpha_n(t) \rightarrow c \alpha_n(c^2 t),\label{symmscale}\\
 \mbox{Global phase shift:} \,&& \alpha_n(t) \rightarrow e^{i\theta} \alpha_n(t),\label{symmshift1}\\
 \mbox{Local phase shift:} \, &&\alpha_n(t) \rightarrow e^{i n \phi} \alpha_n(t),\label{symmshift2}
 \end{eqnarray}
where  $c$, $\theta$, and $\phi$ are real parameters. By the Noether theorem, the latter two symmetries give rise to two  conserved quantities:
\begin{eqnarray}
Q(\alpha) & =\sum\limits_{n=0}^{\infty} (n+1) |\alpha_n|^2,\label{charge}\\
E(\alpha) & =\sum\limits_{n=0}^{\infty} (n+1)^2 |\alpha_n|^2\,. \label{lenergy}
\end{eqnarray}
It is proven in Theorem 1.2 of \cite{BHP18} that $H(\alpha) \leq Q(\alpha)^2$ and the equality is achieved if
and only if $\alpha = c p^n$ for some $c, p \in \mathbb{C}$ with $|p| < 1$.

As is shown in Appendix A of \cite{BHP18} (see also \cite{BBE} for generalizations),
there exists another conserved quantity of the conformal flow (\ref{flow}) in the form
\begin{equation}
\label{Z}
Z(\alpha) = \sum_{n=0}^{\infty} (n+1)(n+2) \bar{\alpha}_{n+1} \alpha_n.
\end{equation}
This quantity is related to another one-parameter group of symmetries:
\begin{eqnarray}
\alpha_n(t) \rightarrow \left[ e^{s D} \alpha(t) \right]_n := \sum_{k=0}^{\infty} \frac{s^k}{k!} \left[ D^k \alpha(t) \right]_n,
\label{symmetry}
\end{eqnarray}
where $s \in \mathbb{R}$ is arbitrary and $D$ is a difference operator given by
\begin{equation}
\label{operator-D}
\left[ D \alpha \right]_n := n \alpha_{n-1} - (n+2) \alpha_{n+1}.
\end{equation}
Note that the difference operator $D$ is obtained from $D f = \{ 2 i (Z-\bar{Z}), f \}$ acting
on any function $f(\alpha,\bar{\alpha})$ on phase space, where
the Poisson bracket is defined by
$$
\{ f, g \} := \sum_{k=0}^{\infty} \frac{1}{2 i (k+1)} \left( \frac{\partial f}{\partial \bar{\alpha}_k} \frac{\partial g}{\partial \alpha_k}
- \frac{\partial f}{\partial \alpha_k} \frac{\partial g}{\partial \bar{\alpha}_k}  \right),
$$
thanks to the symplectic structure $\sum_{n=0}^{\infty} 2i (n+1) d \bar{\alpha}_n \wedge (-d \alpha_n)$
of the conformal flow (\ref{flow}). The factor $(2i)$ in the definition of $D$ is chosen for convenience.
There exists another one-parameter group of symmetries obtained from $\tilde{D} f = \{2(Z + \bar{Z}), f\}$,
which is not going to be used in this paper.

Stationary states of the cubic conformal flow (\ref{flow}) are obtained by the separation of variables
 \begin{equation}\label{stac}
  \alpha_n(t) = A_n e^{-i \lambda t + i n \omega t},
\end{equation}
where the complex amplitudes $A_n$ are time-independent, while parameters $\lambda$ and $\omega$ are real.
Substituting \eqref{stac} into \eqref{flow}, we get a nonlinear system of algebraic equations
for the amplitudes:
\begin{equation}\label{sys}
 (n+1) (\lambda-n \omega) A_n=\sum\limits_{j=0}^{\infty} \sum\limits_{k=0}^{n+j} S_{n j k,n+j-k} \,\bar A_j A_k A_{n+j-k}\,.
\end{equation}
Some particular solutions $\{ A_n \}_{n \in \mathbb{N}}$ of
the stationary system (\ref{sys}) are reviewed below.

\subsection{Single-mode states}

The simplest solutions of  (\ref{sys}) are the \emph{single-mode states} given by
 \begin{equation}\label{1mode}
A_n = c\, \delta_{nN},
 \end{equation}
with $\lambda - N \omega = |c|^2$, where $c\in \mathbb{C}$ is an amplitude and $N \in \mathbb{N}$ is fixed.
Thanks to the transformations (\ref{symmscale}) and (\ref{symmshift1}),
we can set $c = 1$ so that $\lambda - N \omega = 1$.

\subsection{Invariant subspace of stationary states}

As is shown in \cite{bal}, system (\ref{sys}) can be reduced to three
nonlinear equations with the substitution
\begin{equation}
\label{sol-invariant}
A_n = (\beta + \gamma n) p^n,
\end{equation}
where $p, \beta, \gamma$ are complex parameters satisfying the nonlinear algebraic system
\begin{align}
&\frac{- \omega p}{(1+y)^2}=\frac{p}{6}\left(2y |\gamma|^2 + \bar{\beta} \gamma \right),\label{eqp}\\
&\frac{\lambda \gamma}{(1+y)^2}=\frac{\gamma}{6}\left(5|\beta|^2+(18y^2+4y) |\gamma|^2+
(6y-1)\bar{\beta} \gamma + 10 y \bar{\gamma} \beta \right),\label{eqr2}\\
&\frac{\lambda \beta}{(1+y)^2}= \beta \left(|\beta|^2+(6y^2+2y) |\gamma|^2 +2 y \beta \bar{\gamma} \right)+
\gamma \left(2 y |\beta|^2+(4y+2)y^2 |\gamma|^2 +y^2 \bar{\beta} \gamma \right), \label{eqr1}
\end{align}
and we have introduced $y := |p|^2/(1-|p|^2)$ so that $1/(1+y) = 1 - |p|^2$.
It is clear from the decay of the sequence (\ref{sol-invariant}) as $n \to \infty$ that $p$ must be restricted
to the unit disk: $|p|< 1$.

The analysis of solutions to system (\ref{eqp})--(\ref{eqr1})
simplifies if we assume that the parameters $p$, $\beta$, $\gamma$ are real-valued,
which can be done without loss of generality. To see this, note that $p$ can be made
real-valued by the transformation \eqref{symmshift2}. If $\beta = 0$,
then equation (\ref{eqr1}) implies $\gamma = 0$ and hence no nontrivial solutions exist.
If $\beta \neq 0$, it can be made real by the transformation (\ref{symmshift1}).
If $\beta$ is real, then $\gamma$ is real, because equation \eqref{eqp} implies that $\Im(\bar \beta \gamma)=0$.
Thus, it suffices to consider system (\ref{eqp})--(\ref{eqr1}) with $p \in (0,1)$ and  $\beta, \gamma \in \mathbb{R}$.

There exist exactly four families of solutions to the system (\ref{eqp})--(\ref{eqr1}).

\subsubsection{Ground state}

Equation (\ref{eqr2}) is satisfied if $\gamma = 0$. This implies that $\omega = 0$
from equation (\ref{eqp}) and $\lambda (1-p^2)^2 = \beta^2$ from equation (\ref{eqr1})
with $p \in (0,1)$ and $\beta \in \mathbb{R}$. Parameterizing
$\beta = c$ and bringing all together yield the geometric sequence
\begin{equation}
 \label{r1}
A_n = c\, p^n, \quad \lambda = \frac{c^2}{(1-p^2)^2}.
\end{equation}
As $p \to 0$, this solution tends to
the $N=0$ single-mode state (\ref{1mode}). Thanks to the transformation
(\ref{symmscale}), one can set $c = 1-p^2$, then $\lambda = 1$ yields the normalized state (\ref{ground-intro}).

As is explained above, the geometric sequence (\ref{r1}) is a  maximizer of $H(\alpha)$ for fixed $Q(\alpha)$,
therefore, we call it the {\em ground state}.
Nonlinear stability of the ground state has been proven in \cite{BHP18},
where the degeneracy due to the  parameter $p$ has been
controlled with the use of conserved quantities $E(\alpha)$ and $Z(\alpha)$.

\subsubsection{Twisted state}

If $\gamma \neq 0$ but $\omega = 0$, equation (\ref{eqp}) is satisfied with
$\gamma = -\beta/(2y)$. Then, equation (\ref{eqr2}) yields $\lambda = \beta^2 (1+y)^3/(4 y)$,
whereas equation (\ref{eqr1}) is identically satisfied with $p \in (0,1)$ and $\beta \in \mathbb{R}$.
Parameterizing $\beta = -2cp$ and bringing all together yield the \emph{twisted state}
\begin{equation}
\label{twisted-mode}
A_n = c p^{n-1} ((1-p^2)n - 2 p^2), \quad \lambda = \frac{c^2}{(1-p^2)^2}.
\end{equation}
As $p \to 0$, this solution
tends  to the $N=1$ single-mode state (\ref{1mode}).
Thanks to transformation (\ref{symmscale}), one can set $c = 1-p^2$, then $\lambda = 1$.

\subsubsection{Pair of stationary states}

If $\gamma \neq 0$ and $\omega \neq 0$, then $\lambda$ can be eliminated from system
(\ref{eqr2}) and (\ref{eqr1}), which results in  the algebraic equation
\begin{equation}\label{r1-r2}
(2\gamma y +\beta) (12 \gamma^2 y^2 + 6 \beta \gamma y +6 \gamma^2 y + \beta^2 +\beta \gamma)=0.
\end{equation}
The first factor corresponds to the twisted state (\ref{twisted-mode}).
Computing $\gamma$ from the quadratic equation in the second factor yield two roots:
\begin{equation}
\label{branch-solutions-1a}
\gamma = -\frac{(1-p^2) \beta}{12 p^2 (1+p^2)} \left[ 1 + 5 p^2 \mp \sqrt{1 - 14 p^2 + p^4} \right].
\end{equation}
The real roots exist for $1 - 14 p^2 + p^4 \geq 0$ which is true for $p^2 \leq 7 - 4 \sqrt{3}$,
or equivalently, for $p \in [0,2 - \sqrt{3}]$. The pair of solutions
can be parameterized as follows:
\begin{equation}\label{sol75}
 \gamma_{\pm} =  c (1-p^2),\quad \beta_{\pm} = - \frac{1}{2} c(1+5 p^{2} \pm \sqrt{1-14 p^2+p^4}),
 \end{equation}
where $c \in \mathbb{R}$ is arbitrary. Using equations \eqref{eqp} and \eqref{eqr2} again, we get
 \begin{equation}\label{omega-lambda}
 \omega_{\pm}  = \frac{c^2}{12} \, \frac{1+p^2\pm \sqrt{1-14 p^2+p^4}}{1-p^2}, \;
 \lambda_{\pm}  =\frac{c^2}{6}\, \left(\frac{3-4p^2}{1-p^2}
\pm \frac{(3+4p^2)\sqrt{1-14 p^2+p^4}}{(1-p^2)^2}\right)\,.
\end{equation}
As $p \to 0$, the branch with the upper sign converges to the $N=0$ single-mode state
\begin{equation}\label{limit-plus}
p \to 0 : \quad A_n \to - c\, \delta_{n 0},\quad \lambda_+ \to c^2, \quad \omega_+ \to \frac{c^2}{6},
\end{equation}
while the branch with the lower sign converges to the $N=1$ single-mode state
\begin{equation}\label{limit-minus}
p \to 0 : \quad A_n \to \tilde{c} \delta_{n 1}, \quad \lambda_- \to \frac{5}{3} \tilde{c}^2, \quad \omega_- \to  \frac{2}{3} \tilde{c}^2,
\end{equation}
where $\tilde{c} = c p = \mathcal{O}(1)$ has been rescaled as $p \to 0$.
The solutions (\ref{sol-invariant}), (\ref{sol75}), and (\ref{omega-lambda})
have been cast as the stationary states (\ref{min-intro}) and (\ref{omega-lambda-intro}).
By using our bifurcation analysis, we will prove that these stationary states
are local minimizers of $H(\alpha)$ for fixed $Q(\alpha)$ and $E(\alpha)$ for small positive $p$.
We also show numerically that this conclusion remains true for every $p \in (0,2-\sqrt{3})$.

We summarize that the four explicit families (\ref{r1}), (\ref{twisted-mode}), and (\ref{omega-lambda})
form a complete set of stationary states given by (\ref{sol-invariant}).

\subsection{Other stationary states}

Other stationary states were constructed in \cite{bal} from
the generating function $u(t,z)$ given by the power series expansion:
\beq
u(t,z)=\sum_{n=0}^\infty
\alpha_n(t) z^n\,.
\label{defu}
\eeq
The conformal flow (\ref{flow}) can be rewritten for $u(t,z)$ as
the integro-differential equation
\beq
i\del_t\del_z (z u)=\frac1{2\pi i}\oint\limits_{|\zeta|=1} \frac{d\zeta}{\zeta}\, \overline{u(t,\zeta)}
\left(\frac{\zeta u(t,\zeta)-z u(t,z)}{\zeta-z}\right)^2.
\label{GGgen}
\eeq

The stationary solutions expressed by (\ref{stac}) yield the generating function in the form
\beq
u(t,z) = U\left( z e^{i \omega t} \right) e^{-i \lambda t},
\label{stat-generating}
\eeq
where $U$ is a function of one variable given by $U(z) = \sum_{n=0}^{\infty} A_n z^n$.
The family of stationary solutions
expressed by (\ref{sol-invariant}) is generated by the function
\beq
U(z) = \frac{\beta}{1-pz} + \frac{\gamma p z}{(1-pz)^2}.
\label{confansatzu}
\eeq
Additionally, any finite Blaschke product
\begin{equation}\label{blaschke}
U(z) = \prod\limits_{k=1}^N \frac{z-\bar p_k}{1-p_k z}
\end{equation}
yields a stationary state with $\lambda = 1$ and $\omega = 0$
for $p_1,\ldots, p_N \in \mathbb{C}$. If $N = 1$ and $p_1 = p \in (0,1)$,
the generating function (\ref{blaschke}) yields a stationary solution of
the system (\ref{sys}) in the form
\begin{equation}
\label{blaschke-1}
A_0 = -p, \quad A_{n \geq 1} = p^{-1} (1-p^2) p^n, \quad \lambda = 1, \quad \omega = 0.
\end{equation}
This solution can be viewed as a continuation of the $N=1$ single-mode state (\ref{1mode}) in $p \neq 0$.

Another family of stationary states is generated by the function
\begin{equation}\label{N2}
U(z) = \frac{c z^{N}}{1-p^{N+1} z^{N+1}},\qquad \lambda=\frac{c^2}{(1-p^{2N+2})^2}, \quad \omega = 0,
\end{equation}
where $N \in \mathbb{N}$, $p \in (0,1)$, and $c \in \mathbb{R}$ is arbitrary.
When $N = 0$, the function (\ref{N2}) generates the geometric sequence
(\ref{r1}). When $N = 1$, the function (\ref{N2}) generates another stationary solution
\begin{equation}
\label{N2-1}
A_n = c \left\{ \begin{array}{lr} p^{n-1}, & n \;\; {\rm odd}, \\ 0, & n \;\; {\rm even}, \end{array} \right.
\quad \lambda = \frac{c^2}{(1-p^4)^2}, \quad \omega = 0,
\end{equation}
which is a continuation of the $N=1$ single-mode state (\ref{1mode}) in $p \neq 0$.
Thanks to transformation (\ref{symmscale}), one can set $c = 1-p^4$, then $\lambda = 1$.
The new solutions (\ref{blaschke-1}) and (\ref{N2-1}) appear in the local bifurcation analysis
from the $N=1$ single-mode state (\ref{1mode}).

\section{Bifurcations from the lowest eigenmode}

We restrict our attention to the real-valued solutions of the stationary equations (\ref{sys}), which satisfy the system of algebraic equations
\begin{equation}\label{sys-real}
 (n+1) (\lambda-n \omega) A_n=\sum\limits_{j=0}^{\infty} \sum\limits_{k=0}^{n+j} S_{n j k,n+j-k} \,A_j A_k A_{n+j-k}\,.
\end{equation}
Here we study bifurcations of stationary states from the lowest eigenmode given by (\ref{1mode}) with $N = 0$.
Without loss of generality, the scaling transformation (\ref{symmscale}) yields $c = 1$ and $\lambda = 1$.
By setting $A_n = \delta_{n 0} + a_{n}$ with real-valued perturbation $a$, we rewrite
the system (\ref{sys-real}) with $\lambda = 1$ in the perturbative form
\begin{equation}\label{sys-pert}
L(\omega) a + N(a) = 0,
\end{equation}
where $L(\omega)$ is a diagonal operator with entries
\begin{equation}
\label{L-plus}
\left[ L(\omega) \right]_{nn} = \left\{ \begin{array}{ll} 2, & n = 0, \\
n(n+1) \omega - n + 1, & n \geq 1, \end{array} \right.
\end{equation}
and $N(a)$ includes quadratic and cubic nonlinear terms:
\begin{equation}
\label{N-terms}
[N(a)]_n = 2 \sum_{j =0}^{\infty} a_j a_{n+j} + \sum_{k=0}^n a_k a_{n-k} +
\sum\limits_{j=0}^{\infty} \sum\limits_{k=0}^{n+j} S_{njk,n+j-k} \,a_j a_k a_{n+j-k}\,.
\end{equation}
We have the following result on the nonlinear terms.

\begin{lemma}
\label{lem-nonlinear}
Fix an integer $m \geq 2$. If
\begin{equation}
\label{amplitude-zero}
a_{m \ell +1} = a_{m \ell +2} = \dots = a_{m \ell +m-1} = 0, \quad \mbox{for every \;\;} \ell \in \mathbb{N},
\end{equation}
then
\begin{equation}
\label{nonlinear-zero}
[N(a)]_{m \ell +1} = [N(a)]_{m \ell +2} = \dots = [N(a)]_{m\ell + m-1} = 0, \quad \mbox{for every \;\;} \ell \in \mathbb{N}.
\end{equation}
\end{lemma}

\begin{proof}
Let us inspect $[N(a)]_n$ in (\ref{N-terms}) for $n = m \ell + \imath$ with $\ell \in \mathbb{N}$ and
$\imath \in \{1,2,\dots,m-1 \}$. Since $a_j \neq 0$ only if $j$ is multiple of $m$, then $a_{n+j} = 0$
in the first term of (\ref{N-terms}) for all $n$ which are not a multiple of $m$.
Similarly, since $a_k \neq 0$ only if $k$ is multiple of $m$, then $a_{n-k} = 0$
in the second term of (\ref{N-terms}).
Finally, since $a_j \neq 0$ and $a_k \neq 0$ only if $j$ and $k$ are multiple of $m$, then $a_{n+j-k} = 0$
in the third term of (\ref{N-terms}).
Hence, all terms of (\ref{N-terms}) are identically zero for $n = m \ell + \imath$ with $\ell \in \mathbb{N}$ and
$\imath \in \{1,2,\dots,m-1 \}$.
\end{proof}

Bifurcations from the lowest eigenmode are identified by zero eigenvalues
of the diagonal operator $L(\omega)$.

\begin{lemma}
\label{lem-bifurcations}
There exists a sequence of bifurcations at $\omega \in \{ \omega_m \}_{m \in \mathbb{N}_+}$
with
\begin{equation}
\label{bif-points}
\omega_m = \frac{m-1}{m(m+1)}, \quad m \in \mathbb{N}_+,
\end{equation}
where all bifurcation points are simple except for $\omega_2 = \omega_3 = 1/6$.
\end{lemma}

\begin{proof}
The sequence of values (\ref{bif-points}) yield zero diagonal entries in (\ref{L-plus}).
To study if the bifurcation points are simple, we consider solutions of
$\omega_n = \omega_m$ for $n \neq m$. This equation is equivalent to
$mn = m+n+1$, which has only two solutions $(m,n) \in \{ (2,3); (3,2) \}$.
Therefore, all bifurcation points are simple except for the double point
$\omega_2 = \omega_3 = 1/6$.
\end{proof}

The standard Crandall--Rabinowitz theory \cite{CR71}
can be applied to study bifurcation from the simple zero eigenvalue.

\begin{theorem}
Fix $m = 1$ or an integer $m \geq 4$. There exists a unique branch
of solutions $(\omega,A) \in \mathbb{R} \times \ell^2(\mathbb{N})$
to system (\ref{sys-real}) with $\lambda = 1$, which can be parameterized by
small $\epsilon$ such that $(\omega,A)$ is smooth in $\epsilon$ and
\begin{equation}
|\omega - \omega_m| + \sup_{n \in \mathbb{N}} | A_n - \delta_{n 0} - \epsilon \delta_{n m} | \lesssim \epsilon^2.
\label{simple-bif}
\end{equation}
\label{theorem-simple}
\end{theorem}

\begin{proof}
Let us consider the decomposition:
\begin{equation}
\label{decomposition-simple}
\omega = \omega_m + \Omega, \quad a_n = \epsilon \delta_{n m} + b_n, \quad n \in \mathbb{N},
\end{equation}
where $\omega_m$ is defined by (\ref{bif-points}), $\epsilon$ is arbitrary and $b_m = 0$ is set from
the orthogonality condition $\langle b, e_m \rangle = 0$.
Let $L_* = L(\omega_m)$. The system (\ref{sys-pert}) is decomposed into the
invertible part
\begin{equation}
\label{sys-invertible}
[F(\epsilon,\Omega;b)]_n := \left[ L_* \right]_{nn} b_n + n(n+1) \Omega b_n + \left[ N(\epsilon e_m + b) \right]_n = 0, \quad n \neq m,
\end{equation}
and the bifurcation equation
\begin{equation}
\label{sys-bifurcation}
G(\epsilon,\Omega;b) := m(m+1) \Omega + \epsilon^{-1} \left[ N(\epsilon e_m + b) \right]_m = 0.
\end{equation}
Since
\begin{equation}
\label{L-star}
\left[ L_* \right]_{nn} = \left\{ \begin{array}{ll} 2, & n = 0, \\
\frac{(n-m)(nm - m - n - 1)}{m(m+1)}, & n \geq 1, \end{array} \right.
\end{equation}
we have $[L_*]_{nn} \neq 0$ for every $n \neq m$ and $[L_*]_{nn} \to \infty$ as $n \to \infty$.
Therefore, the Implicit Function Theorem can be applied to
$$
F(\epsilon,\Omega;b) : \mathbb{R}^2 \times \ell^{2,2}(\mathbb{N}) \to \ell^2(\mathbb{N}),
$$
where as is defined by (\ref{sys-invertible}), $F$
is smooth in its variables, $F(0,\Omega;0) = 0$ for every $\Omega \in \mathbb{R}$,
$\partial_b F(0,0;0) = L_*$, and
$\| F(\epsilon,\Omega;0) \|_{\ell^2} \lesssim \epsilon^2$.
For every small $\epsilon$ and small $\Omega$, there exists a unique small solution of
system (\ref{sys-invertible}) in $\ell^2(\mathbb{N})$ such that
$\| b \|_{\ell^2} \lesssim \epsilon^2 (1 + |\Omega|)$. Denote this solution by $b(\epsilon,\Omega)$.

Since
\begin{equation}
\label{N-m}
[N(\epsilon e_m)]_m = \epsilon^3 S_{m m m m},
\end{equation}
one power of $\epsilon$ is
canceled in (\ref{sys-bifurcation}) and the Implicit Function Theorem
can be applied to
$$
G(\epsilon,\Omega;b(\epsilon,\Omega)) : \mathbb{R} \times \mathbb{R} \to \mathbb{R},
$$
where as is defined by (\ref{sys-bifurcation}), $G$
is smooth in its variables, $G(0,0;b(0,0)) = 0$, and
$$
\partial_{\Omega} G(0,0;b(0,0)) = m(m+1) \neq 0.
$$
For every small $\epsilon$, there is a unique small root $\Omega$
of the bifurcation equation (\ref{sys-bifurcation}) such that
$|\Omega| \lesssim \epsilon^2$ thanks to (\ref{N-m}).
\end{proof}

\begin{rem}
The unique branch bifurcating from $\omega_1 = 0$ coincides
with the normalized ground state (\ref{ground-intro}) and yields the exact result
$\Omega = 0$. The small parameter $\epsilon$ is defined in terms
of the small parameter $p$ by $\epsilon := p(1-p^2)$.
\end{rem}

\begin{rem}
For the unique branch bifurcating from $\omega_m$ with $m \geq 4$,
we claim that $A_n \neq 0$, $n \in \mathbb{N}$ if and only if $n$ is a multiple of $m$.
Indeed, by Lemma \ref{lem-nonlinear}, the system (\ref{sys-pert}) can be reduced for
a new sequence $\{ a_{m \ell} \}_{\ell \in \mathbb{N}}$ whereas all other elements are identically zero.
By Lemma \ref{lem-bifurcations}, the value
$\omega = \omega_m$ is a simple bifurcation point of this reduced system. By Theorem \ref{theorem-simple},
there exists a unique branch of solutions of this reduced system
with the bounds (\ref{simple-bif}). Therefore, by uniqueness,
the bifurcating solution satisfies the reduction of Lemma \ref{lem-nonlinear}, that is,
$A_n \neq 0$, $n \in \mathbb{N}$ if and only if $n$ is a multiple of $m$.
\end{rem}

\begin{rem}
There are at least three branches bifurcating from the double point $\omega_2 = \omega_3$. One branch satisfies
(\ref{amplitude-zero}) with $\ell = 2$, the other branch satisfies (\ref{amplitude-zero}) with $\ell = 3$,
and the third branch is given by the explicit solution
(\ref{min-intro}) with (\ref{omega-lambda-intro}) for the upper sign.
We show in Theorem \ref{theorem-double} that these are the only branches bifurcating
from the double point.
\end{rem}

As is well-known \cite{Hale}, normal form equations have to be computed in order to study branches
bifurcating from the double zero eigenvalue at the
bifurcation point $\omega_* := \omega_2 = \omega_3 = 1/6$.

\begin{theorem}
Fix $\omega_* := \omega_2 = \omega_3 = 1/6$.
There exist exactly three branches of solutions $(\omega,A) \in \mathbb{R} \times \ell^2(\mathbb{N})$
to system (\ref{sys-real}) with $\lambda = 1$, which can be parameterized by
small $(\epsilon,\mu)$ such that $(\omega,A)$ is smooth in $(\epsilon,\mu)$ and
\begin{equation}
|\omega - \omega_*| + \sup_{n \in \mathbb{N}} | A_n - \delta_{n 0} - \epsilon \delta_{n 2} - \mu \delta_{n 3} |
\lesssim (\epsilon^2 + \mu^2).
\label{double-bif}
\end{equation}
The three branches are characterized by the following:
\begin{itemize}
\item[(i)] $\epsilon = 0$, $\mu \neq 0$;
\item[(ii)] $\epsilon \neq 0$, $\mu = 0$;
\item[(iii)] $\epsilon < 0$, $| \mu - 2 |\epsilon|^{3/2} | \lesssim \epsilon^2$,
\end{itemize}
and the branch (iii) is double degenerate up to the reflection $\mu \mapsto -\mu$.
\label{theorem-double}
\end{theorem}

\begin{proof}
Let us consider the decomposition:
\begin{equation}
\label{decomposition-double}
\omega = \omega_* + \Omega, \quad a_n = \epsilon \delta_{n 2} + \mu \delta_{n 3} + b_n, \quad n \in \mathbb{N},
\end{equation}
where $\omega_* := \omega_2 = \omega_3 = 1/6$, $(\epsilon,\mu)$ are arbitrary, and $b_2 = b_3 = 0$ are set from
the orthogonality condition $\langle b, e_2 \rangle = \langle b, e_3 \rangle = 0$.
The system (\ref{sys-pert}) is decomposed into the invertible part
\begin{equation}
\label{sys-invertible-double}
\left\{ \begin{array}{l}
2 b_0 + \left[ N(\epsilon e_2 + \mu e_3 + b) \right]_0 = 0,  \\
\frac{1}{6} (n-2)(n-3) b_n + n(n+1) \Omega b_n +
\left[ N(\epsilon e_2 + \mu e_3 + b) \right]_n = 0, \quad n \neq \{0,2,3\},
\end{array} \right.
\end{equation}
and the bifurcation equations
\begin{equation}
\label{sys-bifurcation-double}
\left\{ \begin{array}{l}
6 \Omega \epsilon +  \left[ N(\epsilon e_2 + \mu e_3 + b) \right]_2  = 0,\\
12 \Omega \mu + \left[ N(\epsilon e_2 + \mu e_3 + b) \right]_3 = 0.\end{array} \right.
\end{equation}
By the Implicit Function Theorem (see the proof of Theorem \ref{theorem-simple}), there exists a unique map
$\mathbb{R}^3 \ni (\epsilon,\mu,\Omega) \mapsto b \in \ell^2(\mathbb{N})$
for small $(\epsilon,\mu,\Omega)$ such that equations (\ref{sys-invertible-double}) are satisfied and
$\| b \|_{\ell^2} \lesssim (\epsilon^2+\mu^2) (1 + |\Omega|)$. Denote this solution by $b(\epsilon,\mu,\Omega)$.

Compared to the proof of Theorem \ref{theorem-simple}, it is now more difficult to
consider solutions of the system of two algebraic equations (\ref{sys-bifurcation-double}).
In order to resolve the degeneracy of these equations, we have to compute the solution
$b(\epsilon,\mu,\Omega)$ up to the cubic terms in $(\epsilon,\mu)$
under the apriori assumption $|\Omega| \lesssim (\epsilon^2 + \mu^2)$.
Substituting the expansion for $b(\epsilon,\mu,\Omega)$ into the system (\ref{sys-bifurcation-double})
and expanding it up to the quartic terms in $(\epsilon,\mu)$, we will be able to confirm the apriori assumption
$|\Omega| \lesssim (\epsilon^2 + \mu^2)$ and to obtain all solutions of the system (\ref{sys-bifurcation-double})
for $(\epsilon,\mu,\Omega)$ near $(0,0,0)$.

First, we note that if $|\Omega| \lesssim (\epsilon^2 + \mu^2)$, then $b_n = \mathcal{O}(3)$ for every $n \geq 7$, where
$\mathcal{O}(3)$ denotes terms of the cubic and higher order in $(\epsilon,\mu)$.
Thanks to the cubic smallness of $b_n$ for $n \geq 7$, we can
compute $\left[ N(\epsilon e_2 + \mu e_3 + b) \right]_n$ for $n \in \{0,1,4,5,6\}$  up to and
including the cubic order:
\begin{eqnarray*}
\left\{ \begin{array}{l}
\left[ N(\epsilon e_2 + \mu e_3 + b) \right]_0 = 2(\epsilon^2 + \mu^2) + \mathcal{O}(4), \\
\left[ N(\epsilon e_2 + \mu e_3 + b) \right]_1 = 2 \epsilon \mu + 2 \epsilon b_1 + 2 \mu b_4 + 2 \epsilon^2 \mu + \mathcal{O}(4), \\
\left[ N(\epsilon e_2 + \mu e_3 + b) \right]_4 = \epsilon^2 + 2 \epsilon b_6 + 2 \mu b_1 + 3 \epsilon \mu^2 + \mathcal{O}(4), \\
\left[ N(\epsilon e_2 + \mu e_3 + b) \right]_5 = 2 \epsilon \mu + \mathcal{O}(4), \\
\left[ N(\epsilon e_2 + \mu e_3 + b) \right]_6 = \mu^2 + 2 \epsilon b_4 + \mathcal{O}(4).
\end{array} \right.
\end{eqnarray*}
By using the system (\ref{sys-invertible-double}) and the quadratic approximations for $(b_1,b_4,b_6)$,
we obtain $(b_0,b_1,b_4,b_5,b_6)$ up to and including the cubic order:
\begin{eqnarray*}
\left\{ \begin{array}{l}
b_0 = -\epsilon^2 - \mu^2 + \mathcal{O}(4), \\
b_1 = -6 \epsilon \mu + 48 \epsilon^2 \mu + \mathcal{O}(4), \\
b_4 = -3 \epsilon^2 + 30 \epsilon \mu^2 + \mathcal{O}(4), \\
b_5 = -2 \epsilon \mu + \mathcal{O}(4), \\
b_6 = -\frac{1}{2} \mu^2 + 3 \epsilon^3 + \mathcal{O}(4).
\end{array} \right.
\end{eqnarray*}

Next, we compute $\left[ N(\epsilon e_2 + \mu e_3 + b) \right]_n$ for $n \in \{2,3\}$ up to
and including the quartic order:
\begin{eqnarray*}
\left\{ \begin{array}{l}
\left[ N(\epsilon e_2 + \mu e_3 + b) \right]_2 = 4 \epsilon b_0 + 2 \epsilon b_4 + 2 \mu b_1 + 2 \mu b_5 + b_1^2 + 2 b_4 b_6 \\
\phantom{texttext} + 4 \epsilon \mu b_1 + 3 \epsilon^3 + 6 \epsilon \mu^2 + 3 \mu^2 b_4 + \mathcal{O}(5), \\
\left[ N(\epsilon e_2 + \mu e_3 + b) \right]_3 = 2 \epsilon b_1 + 2 \epsilon b_5 + 4 \mu b_0 + 2 \mu b_6 + 2 b_1 b_4 \\
\phantom{texttext} + 2 \epsilon^2 b_1 + 6 \epsilon^2 \mu + 6 \epsilon \mu b_4 + 4 \mu^3 + \mathcal{O}(5).
\end{array} \right.
\end{eqnarray*}
By using the quadratic and cubic approximations for $(b_0,b_1,b_4,b_5,b_6)$, we rewrite
the system (\ref{sys-bifurcation-double}) up to and including the quartic order:
\begin{eqnarray}
\label{nf-eq}
\left\{ \begin{array}{l}
6 \Omega \epsilon = \epsilon (7 \epsilon^2 + 14 \mu^2 - 162 \epsilon \mu^2) + \mathcal{O}(5), \\
12 \Omega \mu  = \mu (14 \epsilon^2 + \mu^2 - 108 \epsilon^3) + \mathcal{O}(5).
\end{array} \right.
\end{eqnarray}
By Lemma \ref{lem-nonlinear}, if $\epsilon = 0$, then $\left[ N(\mu e_3 + b(0,\mu,\Omega)) \right]_2 = 0$,
whereas if $\mu = 0$, then $\left[ N(\epsilon e_2 + b(\epsilon,0,\Omega)) \right]_3 = 0$.
Hence, the system (\ref{nf-eq}) can be rewritten in the equivalent form:
\begin{eqnarray}
\label{nf-eq-better}
\left\{ \begin{array}{l}
6 \Omega \epsilon = \epsilon \left( 7 \epsilon^2 + 14 \mu^2 - 162 \epsilon \mu^2  + \mathcal{O}(4) \right), \\
12 \Omega \mu  = \mu \left( 14 \epsilon^2 + \mu^2 - 108 \epsilon^3 + \mathcal{O}(4) \right).
\end{array} \right.
\end{eqnarray}

We are looking for solutions to the system (\ref{nf-eq-better}) with $(\epsilon,\mu) \neq (0,0)$, because
$b(0,0,\Omega) = 0$ follows from the system (\ref{sys-invertible-double}). There exist
three nontrivial solutions to the system (\ref{nf-eq}):
\begin{itemize}
\item[(I)] $\epsilon = 0$, $\mu \neq 0$, and $\Omega = \frac{1}{12} \mu^2 + \mathcal{O}(4)$;
\item[(II)] $\epsilon \neq 0$, $\mu = 0$, and $\Omega = \frac{7}{6} \epsilon^2 + \mathcal{O}(4)$;
\item[(III)] $\epsilon \neq 0$, $\mu \neq 0$, and
\begin{eqnarray}
\label{nf-eq-red}
\left\{ \begin{array}{l}
6 \Omega = 7 \epsilon^2 + 14 \mu^2 - 162 \epsilon \mu^2 + \mathcal{O}(4), \\
12 \Omega = 14 \epsilon^2 + \mu^2 - 108 \epsilon^3 + \mathcal{O}(4).
\end{array} \right.
\end{eqnarray}
\end{itemize}
All the solutions satisfy the apriori assumption $|\Omega| \lesssim (\epsilon^2 + \mu^2)$.
In order to consider persistence of these solutions in the system (\ref{nf-eq-better}),
we compute the Jacobian matrix:
\begin{equation*}
J(\epsilon,\mu,\Omega) = \left[\begin{matrix} -6 \Omega + 21 \epsilon^2 + 14 \mu^2 - 324 \epsilon \mu^2 &
28 \epsilon \mu - 324 \epsilon^2 \mu \\
28 \epsilon \mu - 324 \epsilon^2 \mu & -12 \Omega + 14 \epsilon^2 + 3 \mu^2 - 108 \epsilon^3 \end{matrix} \right] + \mathcal{O}(4).
\end{equation*}
We proceed in each case as follows:
\begin{itemize}
\item[(I)] The Jacobian is invertible for small $(\epsilon,\mu)$ admitting the expansion in (I).
By the Implicit Function Theorem, there exists a unique continuation of this root
in the system (\ref{nf-eq}). By Lemma \ref{lem-nonlinear} with $m = 3$, the bifurcating solution
corresponds to the reduction with $a_{3j+1} = a_{3j+2} = 0$ for every $j \in \mathbb{N}$,
hence $\epsilon = 0$ persists beyond all orders of the expansion. This yields the solution (i).

\item[(II)] The Jacobian is invertible for small $(\epsilon,\mu)$ admitting the expansion in (II).
By the Implicit Function Theorem, there exists a unique continuation of this root
in the system (\ref{nf-eq}). By Lemma \ref{lem-nonlinear} with $m = 2$, the bifurcating solution
corresponds to the reduction with $a_{2j+1} = 0$ for every $j \in \mathbb{N}$,
hence $\mu = 0$ persists beyond all orders of the expansion.
This yields the solution (ii).

\item[(III)] Eliminating $\Omega$ from the system (\ref{nf-eq-red}) yields the root finding problem:
$$
27 \mu^2 - 324 \epsilon \mu^2 + 108 \epsilon^3 + \mathcal{O}(4) = 0,
$$
which has only two solutions for $(\epsilon,\mu)$ near $(0,0)$ from
the two roots of the quadratic equation $\mu^2 + 4 \epsilon^3 + \mathcal{O}(4) = 0$.
Computing
$$
\det(J(\epsilon,\mu,\Omega)) = 3024 \epsilon^5 + \mathcal{O}(6) \neq 0
$$
verifies that $J(\epsilon,\mu,\Omega)$ is invertible at each of the two roots.
By the Implicit Function Theorem, there exists a unique continuation of each of the two roots
in the system (\ref{nf-eq}). For the root with $\mu > 0$, this yields the solution (iii).
Thanks to the symmetry (\ref{symmshift2}) with $\phi = \pi$, every solution with $\mu > 0$
can be uniquely reflected to the solution with $\mu < 0$ by
the transformation $a_{2k} \mapsto a_{2k}$ and $a_{2k+1} \mapsto - a_{2k+1}$ for $k \in \mathbb{N}$.
Hence the two solutions with $\mu > 0$ and $\mu < 0$ are generated from
the same branch (iii) up to reflection $\mu \mapsto -\mu$.
\end{itemize}
No other branches bifurcate from the point $\omega_* = 1/6$.
\end{proof}

\begin{rem}
\label{remark-state}
The branch (iii) bifurcating from $\omega_* = 1/6$ coincides
with the exact solution (\ref{sol75}) and (\ref{omega-lambda}) for the upper sign.
Indeed, by normalizing $\lambda = 1$ and taking the limit $p \to 0$ as in (\ref{limit-plus}),
we obtain $c = -1 - p^2 + \mathcal{O}(p^4)$, $\beta = 1 + \mathcal{O}(p^4)$, and
$\gamma = -1 + \mathcal{O}(p^4)$. The parameter $\epsilon$ and $\mu$ are related
to each other by means of the small parameter $p$ with the definitions
$\epsilon := -p^2 + \mathcal{O}(p^6)$ and
$\mu := -2 p^3 + \mathcal{O}(p^7)$, hence $\mu^2 + 4 \epsilon^3 = \mathcal{O}(\epsilon^5)$.
\end{rem}

\begin{rem}
Branches (i) and (ii) bifurcating from $\omega_* = 1/6$ can be obtained by the Crandall--Rabinowitz theory \cite{CR71}
by reducing sequences on the constrained subspace of $\ell^2(\mathbb{N})$ by the constraints (\ref{amplitude-zero})
with $\ell = 2$ and $\ell = 3$ respectively.
The zero eigenvalue is simple on the constrained subspace of $\ell^2(\mathbb{N})$, which enables
application of the theory, as it was done in \cite{GGT} in a similar context.
\end{rem}

\section{Variational characterization of the bifurcating states}

Stationary states (\ref{stac}) with parameters $\lambda$ and $\omega$ are critical points of the functional
\begin{equation}\label{K}
K(\alpha) = \frac{1}{2} H(\alpha) - \lambda  Q(\alpha) - \omega \left[ Q(\alpha) - E(\alpha) \right],
\end{equation}
where $H$, $Q$, and $E$ are given by (\ref{Hconf}), (\ref{charge}), and (\ref{lenergy}). Let
$\alpha = A + a + i b$, where $A$ is a real root of the algebraic system \eqref{sys}, whereas $a$
and $b$ are real and imaginary parts of the perturbation.
Because the stationary solution $A$ is a critical point of $K$, the first variation of $K$
vanishes at $\alpha = A$ and the second variation of $K$ at $\alpha = A$
can be written as a quadratic form associated with the Hessian
operator. In variables above, we obtain the quadratic form in the diagonalized form:
\beq
\label{expansions-K}
K(A + a + i b) - K(A) = \langle L_+ a, a \rangle + \langle L_- b, b \rangle + \mathcal{O}(\| a \|^3 + \| b \|^3),
\eeq
where $\langle \cdot, \cdot \rangle$ is the inner product in $\ell^2(\mathbb{N})$ and $\| \cdot \|$ is the induced norm.
After straightforward computations, we obtain the explicit form for the self-adjoint operators
$L_{\pm} : D(L_{\pm}) \to \ell^2(\mathbb{N})$, where $D(L_{\pm}) \subset \ell^2(\mathbb{N})$ is the maximal domain and
$L_{\pm}$ are unbounded operators given by
\begin{eqnarray}
(L_{\pm} a)_n = \sum_{j = 0}^{\infty} \sum_{k = 0}^{n+j} S_{njk,n+j-k} \left[ 2 A_j A_{n+j-k} a_k \pm A_k A_{n+j-k} a_j \right]
- (n+1) (\lambda - n \omega) a_n.
\label{Hessian}
\end{eqnarray}
The following lemma gives variational characterization of the lowest eigenmode
at the bifurcation points $\{ \omega_m \}_{m \in \mathbb{N}_+}$ in Lemma \ref{lem-bifurcations}.

\begin{lemma}
\label{lemma1}
The following is true:
\begin{itemize}
\item For $\omega_1 = 0$, the $N = 0$ single-mode state (\ref{1mode}) is a degenerate saddle point of $K$ with
one positive eigenvalue, zero eigenvalue of multiplicity three, and infinitely many negative eigenvalues bounded away from zero.

\item For $\omega_2 = \omega_3 = 1/6$, the $N = 0$ single-mode state (\ref{1mode}) is a degenerate minimizer of $K$ with
zero eigenvalue of multiplicity five and infinitely many positive eigenvalues bounded away from zero.

\item For $\omega_m$ with $m \geq 4$, the $N = 0$ single-mode state (\ref{1mode}) is a degenerate saddle point of $K$ with
$2(m-2)$ negative eigenvalues, zero eigenvalue of multiplicity three, and infinitely many positive eigenvalues bounded away from zero.
\end{itemize}
\end{lemma}

\begin{proof}
By the scaling transformation (\ref{symmscale}), we take $\lambda = 1$ and $A_n = \delta_{n 0}$,
for which the explicit form (\ref{Hessian}) yields
\beq
\label{Hessian-one-mode}
(L_{\pm} a)_n = (1-n) a_n \pm a_0 \delta_{n 0} + n (n+1) \omega a_n, \quad n \in \mathbb{N}.
\eeq
We note that $L_+ = L(\omega)$ given by (\ref{L-plus}) and $L_-$ is only different from $L_+$ at the first diagonal entry at $n = 0$
(which is $0$ instead of $2$). Because $L_{\pm}$ in (\ref{Hessian-one-mode}) are diagonal, the assertion of the lemma
is proven from explicit computations:
\begin{itemize}
\item If $\omega_1 = 0$, then $\sigma(L_+) = \{ 2,0,-1,-2,\dots\}$ and $\sigma(L_-) =
\{0,0,-1,-2,\dots\}$, which yields the result.

\item If $\omega_* := \omega_2 = \omega_3 = 1/6$, then
$$
[L(\omega_*)]_{nn} = \frac{1}{6} (n-2) (n-3), \quad n \geq 1,
$$
which yields the result.

\item If $\omega_m = (m-1)/(m(m+1))$ with $m \geq 4$, then
$$
[L(\omega_m)]_{nn} = \frac{1}{m(m+1)} (n-m) (mn-m-n-1), \quad n \geq 1.
$$
For $n = 1$ and every $n > m$, $[L(\omega_m)]_{nn} > 0$.
For $2 \leq n \leq m-1$, $[L(\omega_m)]_{nn} < 0$.
For $n = m$, $[L(\omega_m)]_{mm} = 0$. This yields the result.
\end{itemize}
\end{proof}

\begin{rem}
It is shown in Lemma 6.2 of \cite{BHP18} that the normalized ground state (\ref{ground-intro}) bifurcating from $\omega_1$
has the same variational characterization as the $N = 0$ single-mode state, hence it is
a triple-degenerate constrained maximizer of $H$ subject to fixed $Q$.
The triple degeneracy of the family (\ref{r1}) is due to two gauge symmetries (\ref{symmshift1}) and
(\ref{symmshift2}), as well as the presence of the additional parameter $p$.
As is explained in \cite{BBE}, the latter degeneracy is due to the additional symmetry (\ref{symmetry}).
Indeed, by applying $e^{s D}$ with $s \in \mathbb{R}$ to $\alpha_n(t) = \delta_{n0} e^{-it}$
and using the general transformation law derived in \cite{BBE}, we obtain another solution in the form:
$$
\alpha_n(t) = \frac{(\tanh s)^n}{\cosh^2 s} e^{-it},
$$
which coincides with the ground state (\ref{ground-intro}) after the definition $p := \tanh s$.
\end{rem}

\begin{rem}
All branches bifurcating from $\omega_m$ for $m \geq 2$ cannot be obtained by applying $e^{sD}$ with $s \in \mathbb{R}$
to the $N = 0$ single-mode state $\alpha_n(t) = \delta_{n0} e^{-it}$, because the latter state is independent of the values of $\omega$.
\end{rem}

\begin{rem}
All branches bifurcating from $\omega_m$ for $m \geq 4$ have too many negative and positive eigenvalues
and therefore, they represent saddle points of $H$ subject to fixed $Q$ and $E$.
\end{rem}

It remains to study the three branches bifurcating from $\omega_* = \omega_2 = \omega_3 = 1/6$. By Lemma \ref{lemma1},
there is a chance that some of these three branches represent constrained minimizers of $H$ subject to fixed $Q$ and $E$.
One needs to consider how the zero eigenvalue of multiplicity five splits when $\omega \neq \omega_*$ with $|\omega - \omega_*|$ sufficiently small
and how the eigenvalues change under the two constraints of fixed $Q$ and $E$. The following lemma presents
the count of negative eigenvalues of the operators $L_{\pm}$ denoted as $n(L_{\pm})$
at the three branches bifurcating from $\omega_*$.

\begin{lemma}
Consider the three bifurcating branches in Theorem \ref{theorem-double} for
$\omega_* := \omega_2 = \omega_3 = 1/6$. For every $\omega \neq \omega_*$ with $|\omega - \omega_*|$ sufficiently small,
the following is true:
\begin{itemize}
\item[(i)] $n(L_+) = 2$, $n(L_-) = 1$;
\item[(ii)] $n(L_+) = 1$, $n(L_-) = 1$;
\item[(iii)] $n(L_+) = 1$, $n(L_-) = 0$.
\end{itemize}
For each branch, $L_-$ has a double zero eigenvalue, $L_+$ has no zero eigenvalue,
and the rest of the spectrum of $L_+$ and $L_-$ is strictly positive and is bounded away from zero.
\label{lemma2}
\end{lemma}

\begin{proof}
Substituting $\lambda = 1$ and $\omega = \omega_* + \Omega$ into (\ref{Hessian}) yields
\begin{eqnarray}
\nonumber
(L_{\pm} a)_n & = & \sum_{j = 0}^{\infty} \sum_{k = 0}^{n+j} S_{njk,n+j-k} \left[ 2 A_j A_{n+j-k} a_k \pm A_k A_{n+j-k} a_j \right] \\
& \phantom{t} & - 2 a_n + \frac{1}{6} (n-2)(n-3) a_n + n(n+1) \Omega a_n,
\label{Hessian-explicit}
\end{eqnarray}
where $A_n = \delta_{n 0} + \epsilon \delta_{n 2} + \mu \delta_{n 3} + b_n$ with
$\langle b, e_2 \rangle = \langle b, e_3 \rangle = 0$ follows from the decomposition (\ref{decomposition-double}).
The correction terms $(\Omega,b)$ are uniquely defined by parameters $(\epsilon,\mu)$. In what follows,
we consider the three branches in Theorem \ref{theorem-double} separately.

{\bf Case (i): $\epsilon = 0$, $\mu \neq 0$.} Here we have $\Omega = \frac{1}{12} \mu^2 + \mathcal{O}(\mu^4)$,
$b_0 = -\mu^2 + \mathcal{O}(\mu^4)$, $b_1 = b_4 = b_5 = 0$,
$b_6 = -\frac{1}{2} \mu^2 + \mathcal{O}(\mu^4)$, and $b_n = \mathcal{O}(\mu^3)$ for $n \geq 7$.
Substituting these expansions in (\ref{Hessian-explicit}), we obtain
$$
L_{\pm} = L_{\pm}^{(0)} + \mu L_{\pm}^{(1)} + \mu^2 L_{\pm}^{(2)} + \mathcal{O}(\mu^3),
$$
where
\begin{eqnarray*}
(L_{\pm}^{(0)} a)_n & = & \frac{1}{6} (n-2)(n-3) a_n \pm a_{-n}, \\
(L_{\pm}^{(1)} a)_n & = & 2(a_{n+3} + a_{n-3}) \pm 2 a_{3-n}, \\
(L_{\pm}^{(2)} a)_n & = & \frac{1}{12} n (n+1) a_n + 2 (\min(3,n) - 1) a_n - (a_{n+6}+a_{n-6}) \\
& \phantom{t} & \mp 2 a_{-n} \pm \min(3,n,6-n) a_{6-n}.
\end{eqnarray*}
Let us represent the first $7$-by-$7$ matrix block of the operator $L_{\pm}$ and truncate it
by up to and including $\mathcal{O}(\mu^2)$ terms. The corresponding matrix blocks denoted by
$\hat{L}_+$ and $\hat{L}_-$ are given respectively by
$$
\hat{L}_+ = \left[ \begin{array}{ccccccc} 2 - 4 \mu^2 & 0 & 0 & 4 \mu & 0 & 0 & - \mu^2 \\
0 & \frac{1}{3} + \frac{1}{6} \mu^2 & 2 \mu & 0 & 2\mu & \mu^2 & 0 \\
0 & 2 \mu & \frac{5}{2} \mu^2 & 0 & 2 \mu^2 & 2 \mu & 0 \\
4 \mu & 0 & 0 & 8 \mu^2 & 0 & 0 & 2 \mu \\
0 & 2 \mu & 2 \mu^2 & 0 & \frac{1}{3} + \frac{17}{3} \mu^2 & 0 & 0 \\
0 & \mu^2 & 2 \mu & 0 & 0 & 1 + \frac{13}{2} \mu^2 & 0 \\
-\mu^2 & 0 & 0 & 2 \mu & 0 & 0 & 2 + \frac{15}{2} \mu^2
\end{array} \right]
$$
and
$$
\hat{L}_- = \left[ \begin{array}{ccccccc} 0 & 0 & 0 & 0 & 0 & 0 & -\mu^2 \\
0 & \frac{1}{3} + \frac{1}{6} \mu^2 & -2 \mu & 0 & 2\mu & -\mu^2 & 0 \\
0 & -2 \mu & \frac{5}{2} \mu^2 & 0 & -2 \mu^2 & 2 \mu & 0 \\
0 & 0 & 0 & 2 \mu^2 & 0 & 0 & 2 \mu \\
0 & 2 \mu & -2 \mu^2 & 0 & \frac{1}{3} + \frac{17}{3} \mu^2 & 0 & 0 \\
0 & -\mu^2 & 2 \mu & 0 & 0 & 1 + \frac{13}{2} \mu^2 & 0 \\
-\mu^2 & 0 & 0 & 2 \mu & 0 & 0 & 2 + \frac{15}{2} \mu^2
\end{array} \right]
$$

Zero eigenvalue of $L_+^{(0)}$ is double and is associated with the subspace spanned by $\{e_2,e_3\}$.
There are two invariant subspaces of $\hat{L}_+$, one is spanned by $\{ e_0, e_3, e_6 \}$
and the other one is spanned by $\{ e_1,e_2,e_4,e_5\}$. This makes perturbative analysis easier.
For the subspace spanned by $\{ e_0, e_3, e_6 \}$, the eigenvalue problem is given by
$$
\left\{ \begin{array}{l} (2-4 \mu^2) x_0 + 4 \mu x_3 - \mu^2 x_6 = \lambda x_0, \\
4 \mu x_0 + 8 \mu^2 x_3 + 2 \mu x_6 = \lambda x_3,\\
-\mu^2 x_0 + 2 \mu x_3 + (2 + \frac{15}{2} \mu^2) x_6 = \lambda x_6, \end{array} \right.
$$
The small eigenvalue for small $\mu$ is obtained by normalization $x_3 = 1$. Then, we obtain
$$
x_0 = -2 \mu + \mathcal{O}(\mu^3), \quad x_6 = -\mu + \mathcal{O}(\mu^3),
$$
and
\begin{equation}
\label{eigen-1}
\lambda = - 2 \mu^2 + \mathcal{O}(\mu^4).
\end{equation}
For the subspace spanned by $\{ e_1, e_2, e_4, e_5 \}$, the eigenvalue problem is given by
$$
\left\{ \begin{array}{l} (\frac{1}{3} + \frac{1}{6} \mu^2) x_1 + 2 \mu x_2 + 2 \mu x_4 + \mu^2 x_5 = \lambda x_1, \\
2 \mu x_1 + \frac{5}{2} \mu^2 x_2 + 2 \mu^2 x_4 + 2 \mu x_5  = \lambda x_2,\\
2\mu x_1 + 2 \mu^2 x_2 + (\frac{1}{3} + \frac{17}{2} \mu^2) x_4  = \lambda x_4, \\
\mu^2 x_1 + 2 \mu x_2 + (1 + \frac{13}{2} \mu^2) x_5  = \lambda x_5,\end{array} \right.
$$
The small eigenvalue for small $\mu$ is obtained by normalization $x_2 = 1$. Then, we obtain
$$
x_1 = -6 \mu + \mathcal{O}(\mu^3), \quad x_4 = 30\mu^2 + \mathcal{O}(\mu^4), \quad x_5 = -2 \mu + \mathcal{O}(\mu^3),
$$
and
\begin{equation}
\label{eigen-2}
\lambda = -\frac{27}{2} \mu^2 + \mathcal{O}(\mu^3).
\end{equation}
By the perturbation theory, for every $\mu \neq 0$ sufficiently small, $L_+$ has two simple (small) negative eigenvalues.
Other eigenvalues are bounded away from zero for small $\mu$ and by Lemma \ref{lem-bifurcations},
all other eigenvalues of $L_+$ are strictly positive. Hence, $n(L_+) = 2$.

Zero eigenvalue of $L_-^{(0)}$ is triple and is associated with the subspace spanned by $\{e_0,e_2,e_3\}$.
For every $\mu \neq 0$, a double zero eigenvalue of $L_-$ exists due to the two symmetries (\ref{symmshift1}) and (\ref{symmshift2}).
The two eigenvectors for the double zero eigenvalue of $L_-$ are spanned by  $\{ e_0, e_3, e_6, \dots \}$.
There are two invariant subspaces of $\hat{L}_-$, one is spanned by $\{ e_0, e_3, e_6 \}$
and the other one is spanned by $\{ e_1,e_2,e_4,e_5\}$.
Since we only need to compute a shift of the zero eigenvalue of $L_-^{(0)}$, we only consider
the subspace of $\hat{L}_-$ spanned by $\{ e_1,e_2,e_4,e_5\}$.
For this subspace, the eigenvalue problem for $\hat{L}_-$ is given by
$$
\left\{ \begin{array}{l} (\frac{1}{3} + \frac{1}{6} \mu^2) x_1 - 2 \mu x_2 + 2 \mu x_4 - \mu^2 x_5 = \lambda x_1, \\
-2 \mu x_1 + \frac{5}{2} \mu^2 x_2 - 2 \mu^2 x_4 + 2 \mu x_5  = \lambda x_2,\\
2\mu x_1 - 2 \mu^2 x_2 + (\frac{1}{3} + \frac{17}{2} \mu^2) x_4  = \lambda x_4, \\
- \mu^2 x_1 + 2 \mu x_2 + (1 + \frac{13}{2} \mu^2) x_5  = \lambda x_5,\end{array} \right.
$$
The small eigenvalue for small $\mu$ is obtained by normalization $x_2 = 1$. Then, we obtain
$$
x_1 = 6 \mu + \mathcal{O}(\mu^3), \quad x_4 = -30\mu^2 + \mathcal{O}(\mu^4), \quad x_5 = -2 \mu + \mathcal{O}(\mu^3),
$$
and
\begin{equation}
\label{eigen-3}
\lambda = -\frac{27}{2} \mu^2 + \mathcal{O}(\mu^3).
\end{equation}
By the perturbation theory, for every $\mu \neq 0$ sufficiently small, $L_-$ has one simple (small) negative eigenvalue
and the double zero eigenvalue. Other eigenvalues are bounded away from zero for small $\mu$ and by Lemma \ref{lem-bifurcations},
all other eigenvalues of $L_-$ are strictly positive. Hence, $n(L_-) = 1$. \\

{\bf Case (ii): $\epsilon \neq 0$, $\mu = 0$.} Here we have
$\Omega = \frac{7}{6} \epsilon^2 + \mathcal{O}(\epsilon^4)$, $b_0 = -\epsilon^2 + \mathcal{O}(\epsilon^4)$,
$b_1 = b_5 = 0$, $b_4 = -3 \epsilon^2 + \mathcal{O}(\epsilon^4)$, $b_6 = 3 \epsilon^3 + \mathcal{O}(\epsilon^4)$, and
$b_n = \mathcal{O}(\epsilon^4)$ for $n \geq 7$.
Substituting these expansions in (\ref{Hessian-explicit}), we obtain
$$
L_{\pm} = L_{\pm}^{(0)} + \epsilon L_{\pm}^{(1)} + \epsilon^2 L_{\pm}^{(2)} + \epsilon^3 L_{\pm}^{(3)} + \mathcal{O}(\epsilon^4),
$$
where
\begin{eqnarray*}
(L_{\pm}^{(0)} a)_n & = & \frac{1}{6} (n-2)(n-3) a_n \pm a_{-n}, \\
(L_{\pm}^{(1)} a)_n & = & 2(a_{n+2} + a_{n-2}) \pm 2 a_{2-n}, \\
(L_{\pm}^{(2)} a)_n & = & \frac{7}{6} n (n+1) a_n + 2 (\min(2,n) - 1) a_n - 6 (a_{n+4}+a_{n-4}) \\
& \phantom{t} & \mp 2 a_{-n} \pm (\min(2,n,4-n) - 5) a_{4-n}, \\
(L_{\pm}^{(3)} a)_n & = & -2(a_{n+2}+a_{n-2}) + 6 (a_{n+6}+a_{n-6}) \\
& \phantom{t} & - 6 (\min(2,n) + 1) a_{n+2} - 6 (\min(2,n-2) + 1) a_{n-2}\\
& \phantom{t} & \mp 2 a_{2-n} \mp 6 \min(2,n,6-n) a_{6-n}.
\end{eqnarray*}
Let us represent the first $7$-by-$7$ matrix block of the operator $L_{\pm}$ and truncate it
by up to and including $\mathcal{O}(\epsilon^3)$ terms. The corresponding matrix blocks denoted by
$\hat{L}_+$ and $\hat{L}_-$ are given respectively by
{\small $$
\left[ \begin{array}{ccccccc} 2 - 4 \epsilon^2 & 0 & 4 \epsilon - 10 \epsilon^3 & 0 &
-11 \epsilon^2 & 0 & 6 \epsilon^3 \\
0 & \frac{1}{3} + 2 \epsilon + \frac{7}{3} \epsilon^2 - 2 \epsilon^3 & 0 &
2 \epsilon - 4 \epsilon^2 - 14 \epsilon^3 & 0 & -6\epsilon^2 - 6 \epsilon^3 & 0 \\
4 \epsilon - 10 \epsilon^3 & 0 & 6 \epsilon^2 & 0 & 2 \epsilon - 32 \epsilon^3 & 0 & -6 \epsilon^2 \\
0 & 2 \epsilon - 4 \epsilon^2 - 14 \epsilon^3 & 0 & 16 \epsilon^2 - 12 \epsilon^3 & 0 & 2\epsilon -20 \epsilon^3 & 0 \\
-11 \epsilon^2  & 0 & 2 \epsilon - 32 \epsilon^3 & 0 & \frac{1}{3} + \frac{76}{3} \epsilon^2 & 0 & 2\epsilon - 20\epsilon^3 \\
0 & -6\epsilon^2 - 6 \epsilon^3 & 0 & 2 \epsilon -20 \epsilon^3  & 0 & 1 + 37 \epsilon^2 & 0 \\
6 \epsilon^3 & 0 & -6 \epsilon^2 & 0 & 2 \epsilon -20 \epsilon^3 & 0 & 2 + 51\epsilon^2
\end{array} \right]
$$}and{\small $$
\left[ \begin{array}{ccccccc} 0 & 0 & -6 \epsilon^3 & 0 &
-\epsilon^2 & 0 & 6 \epsilon^3 \\
0 & \frac{1}{3} - 2 \epsilon + \frac{7}{3} \epsilon^2 + 2 \epsilon^3 & 0 &
2 \epsilon + 4 \epsilon^2 - 14 \epsilon^3 & 0 & -6\epsilon^2 + 6 \epsilon^3 & 0 \\
-6 \epsilon^3 & 0 & 12 \epsilon^2 & 0 & 2 \epsilon - 8 \epsilon^3 & 0 & -6\epsilon^2 \\
0 & 2 \epsilon + 4\epsilon^2 - 14 \epsilon^3 & 0 & 16 \epsilon^2 + 12 \epsilon^3 & 0 & 2\epsilon -20 \epsilon^3 & 0 \\
- \epsilon^2  & 0 & 2 \epsilon - 8 \epsilon^3 & 0 & \frac{1}{3} + \frac{76}{3} \epsilon^2 & 0 & 2\epsilon - 20\epsilon^3 \\
0 & -6\epsilon^2 + 6 \epsilon^3 & 0 & 2 \epsilon -20 \epsilon^3  & 0 & 1 + 37 \epsilon^2 & 0 \\
6 \epsilon^3 & 0 & -6 \epsilon^2 & 0 & 2 \epsilon -20 \epsilon^3 & 0 & 2 + 51\epsilon^2
\end{array} \right]
$$}

Zero eigenvalue of $L_+^{(0)}$ is double and is associated with the subspace spanned by $\{e_2,e_3\}$.
There are two invariant subspaces of $\hat{L}_+$, one is spanned by $\{ e_0, e_2, e_4, e_6 \}$
and the other one is spanned by $\{ e_1,e_3,e_5\}$.
For the subspace spanned by $\{ e_0, e_2, e_4, e_6\}$, the eigenvalue problem is given by
$$
\left\{ \begin{array}{l} (2-4 \epsilon^2) x_0 + (4 \epsilon - 10 \epsilon^3) x_2 - 11 \epsilon^2 x_4 + 6 \epsilon^3 x_6 = \lambda x_0, \\
(4 \epsilon - 10 \epsilon^3) x_0 + 6 \epsilon^2 x_2 + (2 \epsilon - 32 \epsilon^3) x_4 - 6 \epsilon^2 x_6 = \lambda x_2,\\
-11 \epsilon^2 x_0 + (2 \epsilon - 32 \epsilon^3) x_2 + (\frac{1}{3} + \frac{76}{3} \epsilon^2) x_4 + (2 \epsilon - 20 \epsilon^3) x_6 = \lambda x_4,\\
6 \epsilon^3 x_0 - 6 \epsilon^2 x_2 + (2 \epsilon - 20 \epsilon^3) x_4 + (2 + 51 \epsilon^2) x_6 = \lambda x_6. \end{array} \right.
$$
The small eigenvalue for small $\epsilon$ is obtained by normalization $x_2 = 1$. Then, we obtain
$$
x_0 = -2 \epsilon + \mathcal{O}(\epsilon^3), \quad x_4 = -6 \epsilon + \mathcal{O}(\epsilon^3), \quad
x_6 = 9 \epsilon^2 + \mathcal{O}(\epsilon^4),
$$
and
\begin{equation}
\label{eigen-4}
\lambda = - 14 \epsilon^2 + \mathcal{O}(\epsilon^4).
\end{equation}
For the subspace spanned by $\{ e_1, e_3, e_5 \}$, the eigenvalue problem is given by
$$
\left\{ \begin{array}{l} (\frac{1}{3} + 2 \epsilon + \frac{7}{3} \epsilon^2 - 2 \epsilon^3) x_1
+ (2 \epsilon - 4 \epsilon^2 - 14 \epsilon^3) x_3 - (6 \epsilon^2 + 6 \epsilon^3) x_5 = \lambda x_1, \\
(2 \epsilon - 4 \epsilon^2 - 14 \epsilon^3) x_1
+ (16 \epsilon^2 - 12 \epsilon^3) x_3 + (2 \epsilon - 20 \epsilon^3) x_5 = \lambda x_3, \\
- (6 \epsilon^2 + 6 \epsilon^3) x_1
+ (2 \epsilon - 20 \epsilon^3) x_3 + (1 + 37 \epsilon^2) x_5 = \lambda x_5. \end{array} \right.
$$
The small eigenvalue for small $\epsilon$ is obtained by normalization $x_3 = 1$. Then, we obtain
$$
x_1 = -6 \epsilon + 48 \epsilon^2 + \mathcal{O}(\epsilon^3), \quad
x_5 = -2 \epsilon + \mathcal{O}(\epsilon^3),
$$
and
\begin{equation}
\label{eigen-5}
\lambda = 108 \epsilon^3 + \mathcal{O}(\epsilon^4).
\end{equation}
By the perturbation theory, for every $\epsilon \neq 0$ sufficiently small, $L_+$ has one simple small negative
eigenvalue and one simple small positive eigenvalue.
Other eigenvalues are bounded away from zero for small $\epsilon$ and by Lemma \ref{lem-bifurcations},
all other eigenvalues of $L_+$ are strictly positive. Hence, $n(L_+) = 1$.

Zero eigenvalue of $L_-^{(0)}$ is triple and is associated with the subspace spanned by $\{e_0,e_2,e_3\}$.
For every $\epsilon \neq 0$, a double zero eigenvalue of $L_-$ exists due to the two symmetries (\ref{symmshift1}) and (\ref{symmshift2}).
The two eigenvectors for the double zero eigenvalue of $L_-$ are spanned by  $\{ e_0, e_2, e_4, e_6, \dots \}$.
There are two invariant subspaces of $\hat{L}_-$, one is spanned by $\{ e_0, e_2, e_4, e_6 \}$
and the other one is spanned by $\{ e_1,e_3,e_5\}$.
Since we only need to compute a shift of the zero eigenvalue of $L_-^{(0)}$, we only consider
the subspace of $\hat{L}_-$ spanned by $\{ e_1,e_3,e_5\}$.
For this subspace, the eigenvalue problem for $\hat{L}_-$ is given by
$$
\left\{ \begin{array}{l} (\frac{1}{3} - 2 \epsilon + \frac{7}{3} \epsilon^2 + 2 \epsilon^3) x_1
+ (2 \epsilon + 4 \epsilon^2 - 14 \epsilon^3) x_3 + (-6 \epsilon^2 + 6 \epsilon^3) x_5 = \lambda x_1, \\
(2 \epsilon + 4 \epsilon^2 - 14 \epsilon^3) x_1
+ (16 \epsilon^2 + 12 \epsilon^3) x_3 + (2 \epsilon - 20 \epsilon^3) x_5 = \lambda x_3, \\
(-6 \epsilon^2 + 6 \epsilon^3) x_1
+ (2 \epsilon - 20 \epsilon^3) x_3 + (1 + 37 \epsilon^2) x_5 = \lambda x_5. \end{array} \right.
$$
The small eigenvalue for small $\epsilon$ is obtained by normalization $x_3 = 1$. Then, we obtain
$$
x_1 = -6 \epsilon - 48 \epsilon^2 + \mathcal{O}(\epsilon^3), \quad
x_5 = -2 \epsilon + \mathcal{O}(\epsilon^3),
$$
and
\begin{equation}
\label{eigen-6}
\lambda = -108 \epsilon^3 + \mathcal{O}(\epsilon^4).
\end{equation}
By the perturbation theory, for every $\epsilon \neq 0$ sufficiently small, $L_-$ has one simple (small) negative eigenvalue
and the double zero eigenvalue. Other eigenvalues are bounded away from zero for small $\epsilon$ and by Lemma \ref{lem-bifurcations},
all other eigenvalues of $L_-$ are strictly positive. Hence, $n(L_-) = 1$. \\

{\bf Case (iii): $\epsilon < 0$, $\mu = 2 |\epsilon|^{3/2} + \mathcal{O}(\epsilon^2)$.} Here
we introduce $\delta := (-\epsilon)^{1/2}$ and write $\Omega = \frac{7}{6} \delta^4 + \frac{28}{3} \delta^6 + \mathcal{O}(\delta^8)$,
$b_0 = -\delta^4 - 4 \delta^6 + \mathcal{O}(\delta^8)$, $b_1 = 12 \delta^5 + \mathcal{O}(\delta^7)$,
$b_4 = -3 \delta^4 + \mathcal{O}(\delta^8)$, $b_5 = 4 \delta^5 + \mathcal{O}(\delta^8)$,
$b_6 = -5 \delta^6 + \mathcal{O}(\delta^8)$, and $b_n = \mathcal{O}(\delta^7)$ for $n \geq 7$.
Substituting these expansions in (\ref{Hessian-explicit}), we obtain
$$
L_{\pm} = L_{\pm}^{(0)} + \delta^2 L_{\pm}^{(2)} + \delta^3 L_{\pm}^{(3)} + \delta^4 L_{\pm}^{(4)} + \delta^5 L_{\pm}^{(5)} + \delta^6 L_{\pm}^{(6)}
+ \mathcal{O}(\delta^7),
$$
where
\begin{eqnarray*}
(L_{\pm}^{(0)} a)_n & = & \frac{1}{6} (n-2)(n-3) a_n \pm a_{-n}, \\
(L_{\pm}^{(2)} a)_n & = & -2(a_{n+2} + a_{n-2}) \mp 2 a_{2-n}, \\
(L_{\pm}^{(3)} a)_n & = & -4(a_{n+3} + a_{n-3}) \mp 4 a_{3-n}, \\
(L_{\pm}^{(4)} a)_n & = & \frac{7}{6} n (n+1) a_n + 2 (\min(2,n) - 1) a_n - 6 (a_{n+4}+a_{n-4}) \\
& \phantom{t} & \mp 2 a_{-n} \pm (\min(2,n,4-n) - 5) a_{4-n}, \\
(L_{\pm}^{(5)} a)_n & = & -20(a_{n+1}+a_{n-1})-8(a_{n+5}+a_{n-5}) +4\min(2,n-1)a_{n-1} \\
& \phantom{t} & +4\min(n,2)a_{n+1} \mp 24a_{1-n} \pm 4(\min(2,n,5-n)-1)a_{5-n}, \\
(L_{\pm}^{(6)} a)_n & = & \frac{28}{3} n (n+1) a_n + 8(a_{n+2} + a_{n-2}) -10(a_{n+6} + a_{n-6})+ 6 \min(2,n-2) a_{n-2} \\
& \phantom{t} & + 6 \min(2,n) a_{n+2} +8(\min(3,n)-1) a_{n} \mp 8 a_{-n} \pm 2a_{2-n} \\
& \phantom{t} & \pm 4\min(3,n,6-n) a_{6-n} \pm 6\min(2,n,6-n) a_{6-n}.
\end{eqnarray*}
Let us represent the first $7$-by-$7$ matrix block of the operator $L_{\pm}$ and truncate it
by up to and including $\mathcal{O}(\delta^6)$ terms.  The corresponding matrix blocks denoted by
$\hat{L}_+$ and $\hat{L}_-$ are given respectively by
$$
\hat{L}_+ = \left[ \begin{array}{cccc} 2 - 4 \delta^4 - 16\delta^6 & -44 \delta^5 & -4 \delta^2 + 10 \delta^6 & -8 \delta^3 \\
-44 \delta^5 & \frac{1}{3} - 2 \delta^2 + \frac{7}{3} \delta^4 + \frac{62}{3} \delta^6 & -4\delta^3 - 16\delta^5 & -2 \delta^2 -4 \delta^4 +14 \delta^6 \\
-4 \delta^2 + 10 \delta^6 & -4\delta^3 -16 \delta^5 & 6 \delta^4 +64 \delta^6 & -8 \delta^5 \\
-8 \delta^3 & -2 \delta^2 - 4 \delta^4 + 14 \delta^6 & -8\delta^5 & 16 \delta^4 + 152 \delta^6  \\
-11 \delta^4  & -4 \delta^3 & -2 \delta^2 + 40 \delta^6 & -12 \delta^5 \\
-12 \delta^5 & -6\delta^4 + 10 \delta^6 & -4 \delta^3 & -2 \delta^2 + 20 \delta^6  \\
-10 \delta^6 &-8 \delta^5 & -6 \delta^4 & -4 \delta^3
\end{array} \right.
$$
$$
 \left. \begin{array}{ccccc}
 -11 \delta^4 & -12\delta^5 & -10 \delta^6 \\
-4\delta^3 & -6\delta^4 + 10 \delta^6 & -8\delta^5 \\
-2 \delta^2 + 40 \delta^6 & -4\delta^3 & -6 \delta^4 \\
-12 \delta^5 & -2\delta^2 +20 \delta^6 & -4 \delta^3   \\
\frac{1}{3} + \frac{76}{3} \delta^4 +\frac{608}{3}\delta^6 & -12 \delta^5 & -2\delta^2 + 20\delta^6 \\
-12 \delta^5  & 1 + 37 \delta^4 + 296 \delta^6 & -12 \delta^5 \\
 -2 \delta^2 + 20 \delta^6 & -12 \delta^5  & 2 + 51\delta^2 + 408 \delta^3
\end{array} \right]
$$
and
$$
\hat{L}_- = \left[ \begin{array}{cccc} 20 & 4 \delta^5  & 6 \delta^6 & 0 \\
4 \delta^5 & \frac{1}{3} + 2 \delta^2 + \frac{7}{3} \delta^4 + \frac{50}{3} \delta^6 & 4\delta^3 - 16\delta^5 & -2 \delta^2 +4 \delta^4 +14 \delta^6 \\
6\delta^6 & 4\delta^3 -16 \delta^5 & 12 \delta^4 +64 \delta^6 & -16 \delta^5 \\
0 & -2 \delta^2 + 4 \delta^4 + 14 \delta^6 & -16\delta^5  & 16 \delta^4 + 104 \delta^6  \\
-\delta^4  & -4 \delta^3 & -2 \delta^2 & -12 \delta^5 \\
-4 \delta^5 & -6\delta^4 - 10 \delta^6 & -4 \delta^3  & -2 \delta^2 + 20 \delta^6  \\
-10 \delta^6 &-8 \delta^5  & -6 \delta^4 & -4 \delta^3
\end{array} \right.
$$
$$
 \left. \begin{array}{ccccc}
 - \delta^4 & -4\delta^5 & -10 \delta^6 \\
-4\delta^3 & -6\delta^4 - 10 \delta^6 & -8\delta^5  \\
-2 \delta^2 & -4\delta^3  & -6 \delta^4 \\
-12 \delta^5 & -2\delta^2 +20 \delta^6 & -4 \delta^3 \\
\frac{1}{3} + \frac{76}{3} \delta^4 +\frac{608}{3}\delta^6 & -12 \delta^5  & -2\delta^2 + 20\delta^6 \\
-12 \delta^5 & 1 + 37 \delta^4 + 296 \delta^6 & -12 \delta^5  \\
 -2 \delta^2 + 20 \delta^6 & -12 \delta^5 & 2 + 51\delta^2 + 408 \delta^3
\end{array} \right]
$$

Zero eigenvalue of $L_+^{(0)}$ is double and is associated with the subspace spanned by $\{e_2,e_3\}$.
Since no invariant subspaces of $\hat{L}_+$ exist, we have to proceed with full perturbative expansions.
As a first step, we express $(x_0,x_1,x_4,x_5,x_5)$ for the subspace spanned by $\{e_0, e_1, e_4, e_5, e_6\}$
in terms of $\{ x_2, x_3\}$ for the subspace spanned by $\{e_2,e_3\}$, $\lambda$, and $\delta$. We assume
that $\lambda = \mathcal{O}(\delta^4)$ for the small eigenvalues and neglect terms of the order $\mathcal{O}(\delta^7)$ and higher.
This expansion is given by
$$
\left\{ \begin{array}{l}
x_0 = (2\delta^2 + 32 \delta^6 + \lambda \delta^2)x_2 +4\delta^3 x_3,\\
x_1 = (12\delta^3 + 192 \delta^5)x_2 + (6 \delta^2 + 48 \delta^4 + 240\delta^6 + 18 \lambda\delta^2)x_3, \\
x_4 = (6\delta^2 -312\delta^6 + 18 \lambda \delta^2)x_2 + 120\delta^5 x_3,\\
x_5 = 4\delta^3 x_2 + (2\delta^2 -58\delta^6 + 2 \lambda \delta^2)x_3, \\
x_6 = (9 \delta^4 - \frac{459}{2} \delta^6) x_2 + (2 \delta^3 - 51 \delta^5 - 408 \delta^6) x_3.
\end{array}
\right.
$$
Next, we substitute these expansions to the third and fourth equations of the eigenvalue problem
for $\hat{L}_+$ and again neglect terms of the order $\mathcal{O}(\delta^7)$ and higher:
$$
\left\{ \begin{array}{l}
(-14\delta^4 - \lambda) x_2 -56 \delta^5 x_3 = 0, \\
-56 \delta^5 x_2 + (-8\delta^6 - \lambda) x_3 = 0.
\end{array}
\right.
$$
The reduced eigenvalue problem has two eigenvalues with the expansions
\begin{equation}
\label{eigen-7}
\lambda = -14 \delta^4 + \mathcal{O}(\delta^6)\quad \mbox{\rm when} \quad x_3 = (4 \delta + \mathcal{O}(\delta^3))x_2
\end{equation}
and
\begin{equation}
\label{eigen-8}
\lambda = 216 \delta^6 + \mathcal{O}(\delta^8) \quad \mbox{\rm when} \quad x_2 = (-4 \delta + \mathcal{O}(\delta^3))x_3.
\end{equation}
By the perturbation theory, for every $\delta \neq 0$ sufficiently small, $L_+$ has one simple small negative
eigenvalue and one simple small positive eigenvalue.
Other eigenvalues are bounded away from zero for small $\delta$ and by Lemma \ref{lem-bifurcations},
all other eigenvalues of $L_+$ are strictly positive. Hence, $n(L_+) = 1$.

Zero eigenvalue of $L_-^{(0)}$ is triple and is associated with the subspace spanned by $\{e_0,e_2,e_3\}$.
We proceed again with full perturbative expansions. First, we express $(x_1,x_4,x_5,x_5)$
for the subspace spanned by $\{e_1, e_4, e_5, e_6\}$
in terms of $\{x_0, x_2, x_3\}$ for the subspace spanned by $\{e_0,e_2,e_3\}$, $\lambda$, and $\delta$. We assume
that $\lambda = \mathcal{O}(\delta^6)$ for the small eigenvalues and neglect terms of the order $\mathcal{O}(\delta^7)$ and higher.
This expansion is given by
$$
\left\{ \begin{array}{l}
x_1 = - 12\delta^5 x_0 + (-12\delta^3 + 192 \delta^5)x_2 + (6 \delta^2 - 48 \delta^4+ 240 \delta^6 )x_3, \\
x_4 = 3\delta^4 x_0 + (6\delta^2 - 546\delta^6 )x_2 + 120 \delta^5 x_3,\\
x_5 = 4\delta^5 x_0 + 4\delta^3 x_2 + (2\delta^2 - 58 \delta^6)x_3, \\
x_6 = 8\delta^6 x_0 + 9\delta^4 x_2 + 2\delta^3 x_3.
\end{array}
\right.
$$
Next, we substitute these expansions to the first, third and fourth equations of the eigenvalue problem
for $\hat{L}_-$ and again neglect terms of the order $\mathcal{O}(\delta^7)$ and higher:
$$
\left\{ \begin{array}{l}
\lambda x_0 = 0, \\
\lambda x_2 = 0, \\
(216 \delta^6 - \lambda)x_3 = 0.
\end{array}
\right.
$$
The double zero eigenvalue persists due to two gauge symmetries, whereas one eigenvalue is expanded by
\begin{equation}
\label{eigen-9}
\lambda = 216 \delta^6 + \mathcal{O}(\delta^8).
\end{equation}
By the perturbation theory, for every $\delta \neq 0$ sufficiently small, $L_-$ has one simple (small) negative eigenvalue
and the double zero eigenvalue. Other eigenvalues are bounded away from zero for small $\delta$ and by Lemma \ref{lem-bifurcations},
all other eigenvalues of $L_-$ are strictly positive. Hence, $n(L_-) = 0$.
\end{proof}

In the remainder of this section, we consider the two constraints related to the fixed values of $Q$ and $E$. The constraints
may change the number of negative eigenvalues of the linearization operator $L_+$ constrained
by the following two orthogonality conditions
\begin{equation}
\label{constraint-ground-state}
[X_c]^{\perp} := \left\{ a \in \ell^2(\mathbb{N}) : \quad \langle MA, a \rangle = \langle M^2 A, a \rangle = 0 \right\},
\end{equation}
where $A$ denotes a real-valued solution of system (\ref{sys}) and $M = {\rm diag}(1,2,3,\dots)$.
The constrained space $[X_c]^{\perp}$ is a {\em symplectically orthogonal} subspace  of $\ell^2(\mathbb{N})$
to $X_0 = {\rm span}\{ A, MA \} \subset \ell^2(\mathbb{N})$, the two-dimensional subspace associated
with the double zero eigenvalue of $L_-$ related to the phase rotation symmetries (\ref{symmshift1}) and (\ref{symmshift2}).
Alternatively, the constrained space arises when the perturbation $a$ does not change at the linear approximation
the conserved quantities $Q$ and $E$ defined by (\ref{charge}) and (\ref{lenergy}).
Note that the constraints in (\ref{constraint-ground-state}) are only imposed on the real part
of the perturbation $a$.

Let $\tilde{A}$ denote the stationary state of the stationary equation (\ref{sys}) continued with respect to
two parameters $(\lambda,\omega)$. Let $n(L_+)$ and $z(L_+)$ denote the number of negative and zero eigenvalues of
$L_+$ in $\ell^2(\mathbb{N})$ counted with their multiplicities, where $L_+$ is the linearized operator at $\tilde{A}$.
Let $n_c(L_+)$ and $z_c(L_+)$ denote the number of
negative and zero eigenvalues of $L_+$ constrained in $[X_c]^{\perp}$.
Assume non-degeneracy of the stationary state $\tilde{A}$ in the sense that $z(L_+) = 0$.
By Theorem 4.1 in \cite{pel-book}, we have
\begin{equation}
\label{neg-index}
n_c(L_+) = n(L_+) - p(D) - z(D), \quad z_c(L_+) = z(D),
\end{equation}
where $p(D)$ and $z(D)$ are the number of positive and zero eigenvalues of the $2 \times 2$ matrix
\begin{equation}
\label{D-matrix}
D := \left[ \begin{array}{cc} \frac{\partial \mathcal{Q}}{\partial \lambda} & \frac{\partial \mathcal{Q}}{\partial \omega} \\
\frac{\partial (\mathcal{Q}-\mathcal{E})}{\partial \lambda} & \frac{\partial (\mathcal{Q} - \mathcal{E})}{\partial \omega} \end{array} \right],
\end{equation}
with $\mathcal{Q}(\lambda,\omega) = Q(\tilde{A})$ and $\mathcal{E}(\lambda,\omega) = E(\tilde{A})$ evaluated at the
stationary solution $\tilde{A}$ as a function of the two parameters $(\lambda,\omega)$.

\begin{rem}
\label{remark-pair}
For the pair of stationary states (\ref{min-intro}),
it was computed in \cite{bal} that
\begin{equation}
\label{Q-E}
\mathcal{Q}(\lambda,\omega) = \frac{6}{7} (\lambda + \omega), \quad \mathcal{E}(\lambda,\omega) = 6 \omega,
\end{equation}
where $\lambda$ and $\omega$ are related to the parameters $c$ and $p$ in (\ref{omega-lambda-intro}).
Substituting (\ref{Q-E}) into (\ref{D-matrix}) yields
\begin{equation}
\label{D-explicit}
D = \left[ \begin{array}{cc} 6/7 & 6/7 \\ 6/7 & -36/7 \end{array} \right],
\end{equation}
hence, $D$ has one positive and one negative eigenvalue. If $z(L_+) = 0$ holds,
then $n_c(L_+) = n(L_+) - 1$ and $z_c(L_+) = 0$ by (\ref{neg-index}). We will show in Lemma \ref{lemma3} below that
the same count is true for all three branches of Theorem \ref{theorem-double} bifurcating from $\omega_*$.
\end{rem}

By the scaling transformation (\ref{symmscale}), if the stationary state is given by (\ref{stac}) with real $A$,
then the stationary state is continued with respect to parameter $c > 0$ as
$$
\tilde{\alpha}_n(t) = c A_n e^{-i c^2 \lambda t + i n c^2 \omega t},
$$
hence $\tilde{A} = c A$, $\tilde{\lambda} = c^2 \lambda$, and $\tilde{\omega} = c^2 \omega$.
Substituting these relations into $\mathcal{\tilde{Q}}(\tilde{\lambda},\tilde{\omega}) = Q(\tilde{A})$ and
$\mathcal{\tilde{E}}(\tilde{\lambda},\tilde{\omega}) = E(\tilde{A})$ for $\lambda = 1$ yields
$$
\mathcal{\tilde{Q}}(\tilde{\lambda},\tilde{\omega}) = c^2 \mathcal{Q}_0(\omega) = \tilde{\lambda}
\mathcal{Q}_0(\tilde{\omega} \tilde{\lambda}^{-1})
$$
and
$$
\mathcal{\tilde{E}}(\tilde{\lambda},\tilde{\omega}) = c^2 \mathcal{E}_0(\omega) = \tilde{\lambda}
\mathcal{E}_0(\tilde{\omega} \tilde{\lambda}^{-1}),
$$
where $\mathcal{Q}_0(\omega) = \mathcal{Q}(1,\omega)$ and $\mathcal{E}_0(\omega) = \mathcal{E}(1,\omega)$.
Substituting these representations into (\ref{D-matrix}), evaluating derivatives, and setting $c = 1$ yield
the computational formula
\begin{equation}
\label{D-matrix-simplified}
D = \left[ \begin{array}{cc} \mathcal{Q}_0(\omega) - \omega \mathcal{Q}'_0(\omega) & \mathcal{Q}'_0(\omega) \\
\mathcal{Q}_0(\omega) - \mathcal{E}_0(\omega) - \omega \left[ \mathcal{Q}'_0(\omega) - \mathcal{E}'_0(\omega) \right] &
\mathcal{Q}'_0(\omega) - \mathcal{E}'_0(\omega) \end{array} \right],
\end{equation}
which can be used to compute $D$ for the normalized stationary state $A$ with $\lambda = 1$.
The following lemma gives the variational characterization of
the three branches in Theorem \ref{theorem-double} bifurcating from $\omega_*$
as critical points of $H$ subject to fixed $Q$ and $E$.

\begin{lemma}
Consider the three bifurcating branches in Theorem \ref{theorem-double} for
$\omega_* := \omega_2 = \omega_3 = 1/6$. For every $\omega \neq \omega_*$ with $|\omega - \omega_*|$ sufficiently small,
the following is true:
\begin{itemize}
\item[(i)] The branch is a saddle point of $H$ subject to fixed $Q$ and $E$
with $n_c(L_+) = 1$ and $n(L_-) = 1$;
\item[(ii)] The branch is a saddle point of $H$ subject to fixed $Q$ and $E$
with $n_c(L_+) = 0$ and $n(L_-) = 1$;
\item[(iii)] The branch is a minimizer of $H$ subject to fixed $Q$ and $E$
with $n_c(L_+) = 0$ and $n(L_-) = 0$.
\end{itemize}
The critical points are degenerate only with respect to the two phase rotations
(\ref{symmshift1}) and (\ref{symmshift2}) resulting in $z(L_+) = 0$ and $z(L_-) = 2$.
\label{lemma3}
\end{lemma}

\begin{proof}
{\bf Case (i): $\epsilon = 0$, $\mu \neq 0$.} We compute $\mathcal{Q}_0(\omega)$ and $\mathcal{E}_0(\omega)$
as powers of $\mu$ with the following relation between $\omega$ and $\mu$:
$$
\omega = \frac{1}{6} + \frac{1}{12} \mu^2 + \mathcal{O}(\mu^4).
$$
Then it follows that
\begin{eqnarray*}
\left\{ \begin{array}{l}
\mathcal{Q}_0(\omega) = A_0^2 + 4 A_3^2 + 7 A_6^2 + \dots = 1 + 2 \mu^2 + \mathcal{O}(\mu^4), \\
\mathcal{E}_0(\omega) = A_0^2 + 4^2 A_3^2 + 7^2 A_6^2 + \dots = 1 + 14 \mu^2 + \mathcal{O}(\mu^4),
\end{array} \right.
\end{eqnarray*}
so that
\begin{eqnarray*}
D = \left[ \begin{array}{cc} -3 + \mathcal{O}(\mu^2) & 24 + \mathcal{O}(\mu^2) \\
24 + \mathcal{O}(\mu^2) & -144 + \mathcal{O}(\mu^2) \end{array} \right]
\end{eqnarray*}
has one positive and one negative eigenvalue.
Since $n(L_+) = 2$ and $z(L_+) = 0$ by Lemma \ref{lemma2}, we have $n_c(L_+) = n(L_+) - 1 = 1$
by (\ref{neg-index}). At the same time, $n(L_-) = 1$ and $z(L_-) = 2$ by Lemma \ref{lemma2}. \\

{\bf Case (ii): $\epsilon \neq 0$, $\mu = 0$.} We compute $\mathcal{Q}_0(\omega)$ and $\mathcal{E}_0(\omega)$
as powers of $\epsilon$ with the following relation between $\omega$ and $\epsilon$:
$$
\omega = \frac{1}{6} + \frac{7}{6} \epsilon^2 + \mathcal{O}(\epsilon^4).
$$
Then it follows that
\begin{eqnarray*}
\left\{ \begin{array}{l}
\mathcal{Q}_0(\omega) = A_0^2 + 3 A_2^2 + 5 A_4^2 + \dots = 1 + \epsilon^2 + \mathcal{O}(\epsilon^4), \\
\mathcal{E}_0(\omega) = A_0^2 + 3^2 A_2^2 + 5^2 A_4^2 + \dots = 1 + 7 \epsilon^2 + \mathcal{O}(\epsilon^4),
\end{array} \right.
\end{eqnarray*}
so that
\begin{eqnarray*}
D = \left[ \begin{array}{cc} 6/7 + \mathcal{O}(\epsilon^2) & 6/7 + \mathcal{O}(\epsilon^2) \\
6/7 + \mathcal{O}(\epsilon^2) & -36/7 + \mathcal{O}(\epsilon^2) \end{array} \right]
\end{eqnarray*}
has one positive and one negative eigenvalue. Since
$n(L_+) = 1$ and $z(L_+) = 0$ by Lemma \ref{lemma2}), we have
$n_c(L_+) = n(L_+) - 1 = 0$ by (\ref{neg-index}). At the same time, $n(L_-) = 1$ and $z(L_-) = 2$
by Lemma \ref{lemma2}. \\

{\bf Case (iii): $\epsilon < 0$, $\mu = 2 |\epsilon|^{3/2} + \mathcal{O}(\epsilon^2)$.}
We compute $\mathcal{Q}_0(\omega)$ and $\mathcal{E}_0(\omega)$
as powers of $\epsilon$ with the following relation between $\omega$ and $\epsilon$:
$$
\omega = \frac{1}{6} + \frac{7}{6} \epsilon^2  - \frac{28}{3} \epsilon^3 + \mathcal{O}(\epsilon^4).
$$
Then it follows that
\begin{eqnarray*}
\left\{ \begin{array}{l}
\mathcal{Q}_0(\omega) = A_0^2 + 2 A_1^2 + 3 A_2^2 + 4 A_3^2 + 5 A_4^2 + \dots = 1 + \epsilon^2 + \mathcal{O}(\epsilon^3), \\
\mathcal{E}_0(\omega) = A_0^2 + 2^2 A_1^2 + 3^2 A_2^2 + 4^2 A_3^2 + 5^2 A_4^2 + \dots = 1 + 7 \epsilon^2 + \mathcal{O}(\epsilon^3).
\end{array} \right.
\end{eqnarray*}
The only difference in these expansions compared to the case (ii) is the remainder term as large as $\mathcal{O}(\epsilon^3)$
compared to $\mathcal{O}(\epsilon^4)$. This changes the remainder terms in $D$ to $\mathcal{O}(\epsilon)$ compared to
$\mathcal{O}(\epsilon^2)$ but does not affect the conclusion on $D$. Moreover, we can see that
the computation of $D$ agrees with the exact expression (\ref{D-explicit}) in Remark \ref{remark-pair}. The count $n_c(L_+) = n(L_+) - 1 = 0$
and $n(L_-) = 0$ follows by Lemma \ref{lemma2}.
\end{proof}

\begin{rem}
The presence of conserved quantity $Z(\alpha)$ in (\ref{Z}) does not modify the variational characterization
of the stationary states with $\omega \neq 0$ because $Z(\alpha) = 0$ if $\alpha$ is
the stationary state (\ref{stac}) with $\omega \neq 0$. This follows from the fact that $Z(\alpha)$
is independent of $t$, which is impossible if $Z(\alpha) \neq 0$ and $\omega \neq 0$.
Hence, any stationary state (\ref{stac}) must satisfy the constraint:
$$
\omega \neq 0: \quad \sum_{n=0}^{\infty} (n+1) (n+2) A_n A_{n+1} = 0,
$$
which can be verified for all states in Theorem \ref{theorem-simple} and \ref{theorem-double}
bifurcating from $\omega_m$ for $m \geq 2$.
\end{rem}

\section{Bifurcation from the second eigenmode}

Here we study bifurcations of stationary states in the system of algebraic equations (\ref{sys-real})
from the second eigenmode given by (\ref{1mode}) with $N = 1$.
Without loss of generality, the scaling transformation (\ref{symmscale}) yields
$c = 1$ and $\lambda - \omega = 1$.
By setting $A_n = \delta_{n 1} + a_{n}$ with real-valued perturbation $a$, we rewrite
the system (\ref{sys-real}) with $\lambda = 1 + \omega$ in the perturbative form (\ref{sys-pert}),
where $L(\omega)$ is a block-diagonal operator with the diagonal entries
\begin{equation}
\label{L-plus-N-1}
[L(\omega)]_{nn} = \left\{ \begin{array}{ll} 1-\omega, & n = 0, \\
4, & n = 1, \\
1 + 3 \omega, & n = 2, \\
(n^2-1) \omega - n + 3, & n \geq 3, \end{array} \right.
\end{equation}
and the only nonzero off-diagonal entries $[L(\omega)]_{02} = [L(\omega)]_{20} = 1$, whereas
the nonlinear terms are given by
\begin{eqnarray*}
[N(a)]_n = 2 \sum_{j =0}^{\infty} S_{1nj,n+j-1,1} a_j a_{n+j-1} + \sum_{k=0}^{n+1} S_{1nk,n+1-k} a_k a_{n+1-k}
+ \sum\limits_{j=0}^{\infty} \sum\limits_{k=0}^{n+j} S_{njk,n+j-k} \,a_j a_k a_{n+j-k}\,.
\end{eqnarray*}
We have the following result on the nonlinear terms.

\begin{lemma}
\label{lem-nonlinear-N-1}
Fix an integer $m \geq 2$. If $a_0 = 0$ and
\begin{equation}
\label{amplitude-zero-N-1}
a_{m \ell + 2} = a_{m \ell +3} = \dots = a_{m \ell + m} = 0, \quad \mbox{for every \;\;} \ell \in \mathbb{N},
\end{equation}
then $[N(a)]_0 = 0$ and
\begin{equation}
\label{nonlinear-zero-N-1}
[N(a)]_{m \ell +2} = [N(a)]_{m \ell +3} = \dots = [N(a)]_{m\ell + m} = 0, \quad \mbox{for every \;\;} \ell \in \mathbb{N}.
\end{equation}
\end{lemma}

\begin{proof}
The argument repeats the proof of Lemma \ref{lem-nonlinear}. Under the conditions (\ref{amplitude-zero-N-1}),
every term in $[N(a)]_n$ for $n = m \ell + \imath$ with $\ell \in \mathbb{N}$, and
$\imath \in \{2,3,\dots,m \}$ is inspected and shown to be zero.
\end{proof}

Bifurcations from the second eigenmode are identified by zero eigenvalues
of the diagonal operator $L(\omega)$.

\begin{lemma}
\label{lem-bifurcations-N-1}
There exists a sequence of bifurcations at $\omega \in \{ \omega_m \}_{m \in \mathbb{N}_+}$
with $\omega_1 = 0$, $\omega_2 = 2/3$, and
\begin{equation}
\label{bif-points-N-1}
\omega_m = \frac{m-3}{m^2-1}, \quad m \in \{3,4,\dots\}.
\end{equation}
All bifurcation points are simple except for the three double points
$\omega_1 = \omega_3 = 0$, $\omega_4 = \omega_{11} = 1/15$, and
$\omega_5 = \omega_7 = 1/12$.
\end{lemma}

\begin{proof}
The diagonal terms $[L(\omega)]_{nn}$ for $n \in \{3,4,\dots\}$
vanish at the sequence (\ref{bif-points-N-1}). In addition, the double block
for $n = 0$ and $n = 2$ has zero eigenvalues if and only if $\omega_1 = 0$
or $\omega_2 = 2/3$. Therefore, $\omega_1 = \omega_3 = 0$ is a double bifurcation point.
To study other double bifurcation points, we consider solutions of
$\omega_m = \omega_2 = 2/3$ for $m \geq 3$ and $\omega_n = \omega_m$ for $n \neq m \geq 4$.
Equation $\omega_m = \omega_2 = 2/3$ is equivalent to $2m^2 - 3m + 7 = 0$, which has no integer
solutions. Equation $\omega_n = \omega_m$ for $n \neq m$
is equivalent to $mn = 3(m+n)-1$, which can be solved for $m$ in terms of $n$
$$
m = M(n) := \frac{3n-1}{n-3} = 3 + \frac{8}{n-3}.
$$
Since the right-hand side is monotonically decreasing in $n$ and $M(12) < 4$,
we can find all integer solutions for $n$ in the range from $4$ to $11$.
There exists only two pairs of integer solutions in this range,
which give two double points $\omega_4=\omega_{11}=1/15$ and  $\omega_5=\omega_7=1/12$.
All the remaining bifurcation points are simple.
\end{proof}

Simple bifurcation points can be investigated similarly to the proof of Theorem \ref{theorem-simple}.
This yields the following theorem.

\begin{theorem}
Fix $m = 2$. There exists a unique branch
of solutions $(\omega,A) \in \mathbb{R} \times \ell^2(\mathbb{N})$
to system (\ref{sys-real}) with $\lambda - \omega = 1$, which can be parameterized by
small $\epsilon$ such that $(\omega,A)$ is smooth in $\epsilon$ and
\begin{equation}
|\omega - \omega_2| + \sup_{n \in \mathbb{N}} | A_n - \delta_{n 1} - \epsilon (3 \delta_{n 0} - \delta_{n 2}) | \lesssim \epsilon^2.
\label{simple-bif-N-1}
\end{equation}
Fix an integer $m \geq 6$ with $m \neq 7$ and $m \neq 11$. There exists a unique branch
of solutions $(\omega,A) \in \mathbb{R} \times \ell^2(\mathbb{N})$
to system (\ref{sys-real}) with $\lambda - \omega = 1$, which can be parameterized by
small $\epsilon$ such that $(\omega,A)$ is smooth in $\epsilon$ and
\begin{equation}
|\omega - \omega_m| + \sup_{n \in \mathbb{N}} | A_n - \delta_{n 1} - \epsilon \delta_{n m} | \lesssim \epsilon^2.
\label{simple-bif-N-2}
\end{equation}
\label{theorem-simple-N-1}
\end{theorem}

\begin{proof}
The proof of the second assertion repeats the proof of Theorem \ref{theorem-simple} verbatim.
The proof of the first assertion is based on the block-diagonalization of the singular
matrix for $[L(\omega_2)]_{jk}$ with $j,k \in \{0,2\}$:
$$
\left[ \begin{array}{cc} 1/3 & 1 \\ 1 & 3 \end{array} \right].
$$
The null space is spanned by the vector $(3,-1)^T$ and the vector $b \in \ell^2(\mathbb{N})$
in the decomposition (\ref{decomposition-simple}) must satisfy the constraint $3b_0 - b_2 = 0$.
The rest of the proof repeats the proof of Theorem \ref{theorem-simple} after a simple
observation that $[N(\epsilon (3e_0 - e_2))]_{0,2} = \mathcal{O}(\epsilon^3)$ as $\epsilon \to 0$.
\end{proof}

\begin{rem}
\label{remark-state-minus}
The unique branch bifurcating from $\omega_2 = 2/3$ coincides
with the exact solution (\ref{sol75})  and
(\ref{omega-lambda}) for the lower sign.
Indeed, by normalizing $\lambda - \omega = 1$ and taking the limit $p \to 0$ as in (\ref{limit-minus}),
we obtain $c = p^{-1} + \mathcal{O}(p)$, $\beta = -6p + \mathcal{O}(p^3)$,
and $\gamma = p^{-1} + \mathcal{O}(p)$. The small parameter $\epsilon$ is defined in terms of the small parameter $p$
by $\epsilon := -2p +\mathcal{O}(p^3)$.
\end{rem}

The three double bifurcation points in Lemma \ref{lem-bifurcations-N-1} have to be checked separately.
In order to characterize branches bifurcating from the double point $\omega_5 = \omega_7 = 1/12$,
we note the following symmetry. If $u(t,z)$ is a generating function for the conformal flow (\ref{flow})
given by the power series (\ref{defu}), so is $z u(t,z^2)$. If
$u(t,z)$ is a stationary state in the form (\ref{stat-generating})
with $\{A_n\}_{n \in \mathbb{N}}$ satisfying the system (\ref{sys-real})
with parameters $(\lambda,\omega)$, then the transformed state
\begin{equation}\label{ztoz2}
\tilde{A}_n = \begin{cases} A_{m}, \quad n=2m+1, \\ 0, \quad n=2m,
\end{cases} \quad n \in \mathbb{N}
\end{equation}
also satisfies the system (\ref{sys-real}) with parameters
$$
\tilde{\lambda} = \lambda+\frac{\omega}{2}, \quad
\tilde{\omega} = \frac{\omega}{2}.
$$
By Theorem \ref{theorem-double}, three branches of solutions bifurcate from the
lowest eigenmode at the double bifurcation point $\omega_2=\omega_3 =1/6$. Applying
the transformation (\ref{ztoz2}) yields three branches
bifurcating from the second eigenmode at the double bifurcation point $\omega_5=\omega_7=1/12$.
By computing the normal form similar to the proof of Theorem \ref{theorem-double},
we have checked that no other solutions bifurcate from this double point.
The corresponding computations are omitted here, while the result is
formulated in the following theorem.

\begin{theorem}
Fix $\omega_* := \omega_5 = \omega_7 = 1/12$.
There exist exactly three branches of solutions $(\omega,A) \in \mathbb{R} \times \ell^2(\mathbb{N})$
to system (\ref{sys-real}) with $\lambda - \omega = 1$, which can be parameterized by
small $(\epsilon,\mu)$ such that $(\omega,A)$ is smooth in $(\epsilon,\mu)$ and
\begin{equation}
|\omega - \omega_*| + \sup_{n \in \mathbb{N}} | A_n - \delta_{n 1} - \epsilon \delta_{n 5} - \mu \delta_{n 7} |
\lesssim (\epsilon^2 + \mu^2).
\label{double-bif-N-1}
\end{equation}
The three branches are characterized by the following:
\begin{itemize}
\item[(i)] $\epsilon = 0$, $\mu \neq 0$;
\item[(ii)] $\epsilon \neq 0$, $\mu = 0$;
\item[(iii)] $\epsilon < 0$, $| \mu - 2 |\epsilon|^{3/2} | \lesssim \epsilon^2$,
\end{itemize}
and the branch (iii) is double degenerate up to the reflection $\mu \mapsto -\mu$.
\label{theorem-double-N-1}
\end{theorem}

Branches bifurcating from the double point $\omega_4 = \omega_{11} = 1/15$ can be investigated
by computing the normal form. The following theorem represents the main result.

\begin{theorem}
Fix $\omega_* := \omega_4 = \omega_{11} = 1/15$.
There exist exactly two branches of solutions $(\omega,A) \in \mathbb{R} \times \ell^2(\mathbb{N})$
to system (\ref{sys-real}) with $\lambda - \omega = 1$, which can be parameterized by
small $(\epsilon,\mu)$ such that $(\omega,A)$ is smooth in $(\epsilon,\mu)$ and
\begin{equation}
|\omega - \omega_*| + \sup_{n \in \mathbb{N}} | A_n - \delta_{n 1} - \epsilon \delta_{n 4} - \mu \delta_{n 11} |
\lesssim (\epsilon^2 + \mu^2).
\label{double-bif-N-2}
\end{equation}
The two branches are characterized by the following:
\begin{itemize}
\item[(i)] $\epsilon = 0$, $\mu \neq 0$;
\item[(ii)] $\epsilon \neq 0$, $\mu = 0$.
\end{itemize}
\label{theorem-double-N-2}
\end{theorem}

\begin{proof}
The proof follows the computations of normal form in Theorem \ref{theorem-double}.
For the double bifurcation point $\omega_*$,
we write the decomposition
$$
\omega = \omega_* + \Omega, \quad a_n = \epsilon \delta_{n 4} + \mu \delta_{n 11} + b_n, \quad n \in \mathbb{N},
$$
where $(\epsilon,\mu)$ are arbitrary and $b_4 = b_{11} = 0$ are set from
the orthogonality condition $\langle b, e_4 \rangle = \langle b, e_{11} \rangle = 0$.
By performing routine computations and expanding the bifurcation equations at $n = 4$ and $n = 11$
up to and including the cubic order, we obtain
\begin{eqnarray}
\label{order3N1m411}
\left\{ \begin{array}{l}
15 \Omega \epsilon = \epsilon \left( 7 \epsilon^2 + 14 \mu^2 + \mathcal{O}(3) \right), \\
120 \Omega \mu  = \mu \left(14 \epsilon^2 + \frac{56}{17} \mu^2 + \mathcal{O}(3) \right).
\end{array} \right.
\end{eqnarray}
We are looking for solutions to the system (\ref{order3N1m411}) with $(\epsilon,\mu) \neq (0,0)$. There exist
two nontrivial solutions to the system (\ref{order3N1m411}):
\begin{itemize}
\item[(I)] $\epsilon =  0$, $\mu \neq 0$, and $\Omega = \frac{7}{255} \mu^2 + \mathcal{O}(3)$;
\item[(II)] $\epsilon \neq 0$, $\mu = 0$, and $\Omega = \frac{7}{15} \epsilon^2 + \mathcal{O}(3)$;
\end{itemize}
whereas no solution exists with both $\epsilon \neq 0$ and $\mu \neq 0$.
The Jacobian matrix of the system (\ref{order3N1m411}) is given by
\begin{equation*}
J(\epsilon,\mu,\Omega) = \left[\begin{matrix} -15 \Omega + 21 \epsilon^2 + 14 \mu^2 &
28 \epsilon \mu \\ 28 \epsilon \mu & -120 \Omega + 14 \epsilon^2 + \frac{168}{17} \mu^2 \end{matrix} \right] + \mathcal{O}(3).
\end{equation*}
For both branches (I) and (II), the Jacobian is invertible for small $(\epsilon,\mu)$, hence
the two solutions are continued uniquely with respect to parameters $(\epsilon,\mu)$ and yield branches (i) and (ii).
By Lemma \ref{lem-nonlinear-N-1} with $m = 10$, branch (i) corresponds to the reduction (\ref{amplitude-zero-N-1})
with $m = 10$, hence $\epsilon = 0$ persists beyond all orders of the expansion.
By Lemma \ref{lem-nonlinear-N-1} with $m = 3$, branch (ii) corresponds to the reduction (\ref{amplitude-zero-N-1}) with $m = 3$,
hence $\mu = 0$ persists beyond all orders of the expansion.
\end{proof}

\begin{rem}
Branch (ii) of Theorem \ref{theorem-double-N-2} can be obtained by
the symmetry transformation \eqref{ztoz2} from the branch bifurcating from lowest eigenstate at $\omega_5 = 2/15$.
\end{rem}

For the remaining double point $\omega_1 = \omega_3 = 0$, bifurcation of stationary state is more complicated. If we compute
the normal form up to and including the cubic order, we obtain a trivial normal form, which is satisfied
identically if $\omega = 0$. This outcome of the normal form computations suggests that
there exists a two-parameter family of solutions $A \in \ell^2(\mathbb{N})$ to the system (\ref{sys-real}) with $\lambda = 1$ and $\omega = 0$.
Indeed, the two eigenvectors for the null space of $L(0)$ are given by $e_0 - e_2$ and $e_3$ so that
we can introduce the decomposition
$$
a_n = \epsilon (\delta_{n0} - \delta_{n2}) + \mu \delta_{n3} + b_n, \quad n \in \mathbb{N},
$$
subject to the orthogonality conditions $b_0 - b_2 = 0$ and $b_3 = 0$. By computing power expansions
for small $(\epsilon,\mu)$ with MAPLE, we can extend it to any polynomial order with the first terms
given by
\begin{eqnarray*}
\left\{ \begin{array}{l}
A_0 =  \epsilon + \epsilon  \mu + \left(\epsilon  \mu ^2-\frac{\epsilon ^3}{2}\right)+ \left(\epsilon ^3 \mu +2 \epsilon
 \mu ^3\right)+ \left(-\frac{5 \epsilon ^5}{4}+\frac{3 \epsilon ^3 \mu ^2}{2}+3 \epsilon  \mu ^4\right) + \mathcal{O}(6), \\
A_1 =  1-\left(\epsilon ^2+\mu ^2\right) + \epsilon ^2 \mu + + \left(8 \epsilon ^4 \mu +3 \epsilon^2 \mu ^3\right)+
\left(-\epsilon ^4-6 \epsilon ^2 \mu ^2-2 \mu ^4\right) + \mathcal{O}(6),\\
A_2 =  - \epsilon + \epsilon  \mu + \left(-\frac{\epsilon ^3}{2}+\epsilon  \mu ^2\right)+ \left(\epsilon ^3 \mu +2 \epsilon
 \mu ^3\right)+ \left(-\frac{5 \epsilon ^5}{4}+\frac{3 \epsilon ^3 \mu ^2}{2}+3 \epsilon  \mu ^4\right) + \mathcal{O}(6),\\
A_3 = \mu, \\
A_4 = -2  \epsilon  \mu + \left(\epsilon ^3+\epsilon  \mu ^2\right)- \left(4 \epsilon ^3 \mu +\epsilon  \mu ^3\right)+
 \left(\frac{5 \epsilon ^5}{2}+\frac{11 \epsilon ^3 \mu ^2}{2}+3 \epsilon  \mu ^4\right) + \mathcal{O}(6),\\
A_5 = \mu ^2+ \epsilon ^2 \mu + \left(3 \epsilon ^4 \mu +\epsilon ^2 \mu ^3\right)+ \left(-\epsilon ^4+\epsilon ^2
\mu ^2+\mu ^4\right) + \mathcal{O}(6), \\
A_6 =  -3  \epsilon  \mu ^2+ \left(2 \epsilon ^3 \mu +\epsilon  \mu ^3\right)- \left(\frac{21}{2} \epsilon ^3 \mu ^2+5 \epsilon
 \mu ^4\right) + \mathcal{O}(6).
 \end{array} \right.
\end{eqnarray*}
The three explicit solutions (\ref{twisted-mode}), (\ref{blaschke-1}), and (\ref{N2-1})
are particular solutions of the two-parameter family of stationary states for small $p$. Indeed,
the twisted state (\ref{twisted-mode}) corresponds to
$\epsilon = -2p + \mathcal{O}(p^3)$ and $\mu = 3 p^2 + \mathcal{O}(p^4)$,
the Blaschke state corresponds to $\epsilon = -p + \mathcal{O}(p^3)$ and $\mu = p^2 + \mathcal{O}(p^4)$,
and the additional state (\ref{N2-1}) corresponds to $\epsilon = 0$ and $\mu = p^2 + \mathcal{O}(p^6)$.

\begin{rem}
We have shown by using the general transformation law derived in \cite{BBE} that
the twisted state (\ref{twisted-mode}) can be obtained from the $N=1$ single-mode state (\ref{1mode})
by applying $e^{sD}$ with $s \in \mathbb{R}$ in the symmetry transformation (\ref{symmetry})
with $p = \tanh s$. In the present time, we do not know how to obtain the two-parameter branch 
of the stationary states above from the $N = 1$ single-mode state (\ref{1mode}).
\end{rem}

\section{Variational characterization of the bifurcating states}

We shall now give variational characterization of the bifurcating states from the second eigenmode.
The second variation of the action functional $K(\alpha)$ in (\ref{K}) is given by the quadratic forms
in (\ref{expansions-K}) with the self-adjoint operators
$L_{\pm} : D(L_{\pm}) \to \ell^2(\mathbb{N})$ given by (\ref{Hessian}).
The following lemma gives variational characterization of the second eigenmode
at the bifurcation points $\{ \omega_m \}_{m \in \mathbb{N}_+}$ in Lemma \ref{lem-bifurcations-N-1}.

\begin{lemma}
\label{lemma1-N-1}
The following is true:
\begin{itemize}
\item For $\omega_1 = \omega_3 = 0$, the $N = 1$ single-mode state (\ref{1mode}) is a degenerate saddle point of $K$ with
three positive eigenvalues, zero eigenvalue of multiplicity three, and infinitely many negative eigenvalues bounded away from zero.

\item For $\omega_2 = 2/3$ and $\omega_6 = 3/35$, the $N = 1$ single-mode state (\ref{1mode}) is a degenerate minimizer of $K$ with
zero eigenvalue of multiplicity three and infinitely many positive eigenvalues bounded away from zero.

\item For $\omega_m$ with $m \geq 4$ and $m \neq 6$, the $N = 1$ single-mode state (\ref{1mode}) is
a degenerate saddle point of $K$ with an even number of negative eigenvalues, zero eigenvalue of odd multiplicity,
and infinitely many positive eigenvalues bounded away from zero.
\end{itemize}
\end{lemma}

\begin{proof}
By the scaling transformation (\ref{symmscale}), we take $\lambda - \omega = 1$ and $A_n = \delta_{n 1}$,
for which the explicit form (\ref{Hessian}) yields
\beq
\label{Hessian-one-mode-N-1}
(L_{\pm} a)_n = [1-n + 2 \min(n,1)] a_n \pm [ 1 + \min(n,1,2-n) a_{2-n} + (n^2 - 1) \omega a_n, \quad n \in \mathbb{N}.
\eeq
We note that $L_+ = L(\omega)$ given by (\ref{L-plus-N-1}) and
$L_-$ is only different from $L_+$ at the diagonal entry at $n = 1$
(which is $0$ instead of $4$) and for the off-diagonal entries at $n = 0$ and $n = 2$
(which are $-1$ instead of $+1$).

For $L(\omega)$, the $2 \times 2$ block  at $n = 0$ and $n = 2$,
$$
\left[ \begin{array}{cc} 1 - \omega & 1 \\ 1 & 1 + 3 \omega \end{array} \right],
$$
is positive definite for all $\{ \omega_m \}_{m \in \mathbb{N}_+}$ with a simple zero eigenvalue
only arising for $\omega_1 = \omega_3 = 0$ and $\omega_2 = 2/3$. The diagonal entries of $L(\omega)$
for $n \geq 3$ are given by
$$
(n^2-1) \omega + 3 - n, \quad n \geq 3.
$$
These entries are negative if $\omega_1 = \omega_3 = 0$ with a simple zero eigenvalue and
strictly positive if $\omega_2 = 2/3$. For $\omega_m = (m-3)/(m^2-1)$ with $m \geq 4$,
the diagonal entries are simplified in the form:
$$
(n^2-1) \omega_m + 3 - n = \frac{(n-m) (nm - 3n - 3m + 1)}{m^2 - 1}, \quad n \geq 3.
$$
For $m = 4$ and $m = 11$, two eigenvalues are zero, six are negative, and all others are positive.
For $m = 5$ and $m = 7$, two eigenvalues are zero, one is negative, and all others are positive.
For $m = 6$, one eigenvalue is zero and all others are positive. For $m = 8, 9, 10$ and $m \geq 12$,
one eigenvalue is zero, finitely many are negative, and all others are positive.

Multiplying these counts for $L_+ = L(\omega)$ by a factor of $2$ due to the matrix operator $L_-$
and adding an additional zero entry of $L_-$ at $n = 1$ yields the assertion of the lemma.
\end{proof}

Among all stationary states bifurcating from the second eigenmode, we shall only consider the potential minimizers
of energy $H$ subject to fixed $Q$ and $E$. By Lemma \ref{lemma1-N-1}, this includes only two
branches bifurcating from $\omega_2$ and $\omega_6$. In both cases, we are able to
compute the number of negative eigenvalues of the operators $L_{\pm}$ denoted by $n(L_{\pm})$.

\begin{lemma}
Consider the two bifurcating branches in Theorem \ref{theorem-simple-N-1} for
$\omega_2$ and $\omega_6$. For every small $\epsilon \neq 0$,
we have $n(L_+) = 1$ and $n(L_-) = 0$. For each branch, $L_-$ has a double zero eigenvalue,
$L_+$ has no zero eigenvalue, and the rest of the spectrum of $L_+$ and $L_-$ is strictly positive
and is bounded away from zero.
\label{lemma2-N-1}
\end{lemma}

\begin{proof}
By the second item of Lemma \ref{lemma1-N-1}, the corresponding operators $L_+$ and $L_-$ at the bifurcation point
$\epsilon = 0$ have respectively the simple and double zero eigenvalue,
whereas the rest of their spectra is strictly positive and is bounded away from zero.
By the two symmetries (\ref{symmshift1}) and (\ref{symmshift2}), the double zero eigenvalue
of $L_-$ is preserved for $\epsilon \neq 0$
and the assertion of the lemma for $L_-$ follows by the perturbation theory.

On the other hand, the simple zero eigenvalue of $L_+$ is not preserved for $\epsilon \neq 0$
and will generally shift to either negative or positive values. We will show that it shifts
to the negative values for $\epsilon \neq 0$, hence $n(L_+) = 1$ in both cases and the assertion of the lemma
for $L_+$ follows by the perturbation theory.

For $\omega_2$, the Lyapunov--Schmidt decomposition of Theorem \ref{theorem-simple-N-1} yields
power expansion
$$
\omega = \frac{2}{3} - \frac{7}{3} \epsilon^2 + \mathcal{O}(\epsilon^4)
$$
and $A_0 = 3 \epsilon + \mathcal{O}(\epsilon^3)$, $A_1 = 1 - 4 \epsilon^2 + \mathcal{O}(\epsilon^4)$,
$A_2 = -\epsilon + \mathcal{O}(\epsilon^3)$, $A_3  = \frac{3}{4} \epsilon^2 + \mathcal{O}(4)$,
and $A_n = \mathcal{O}(\epsilon^3)$ for $n \geq 4$. Similarly to the proof of Lemma \ref{lemma2},
we compute the $4$-by-$4$ block of the operator $L_+$ at $n \in \{0,1,2,3\}$ and truncate it up
to and including $\mathcal{O}(\epsilon^2)$ terms. The corresponding
matrix block denoted by by $\hat{L}_+$ is given by
$$
\hat{L}_+ = \left[ \begin{array}{cccc}
\frac{1}{3}+\frac{46}{3}\epsilon^2 & 10\epsilon & 1-\frac{37}{2}\epsilon^2 & -2\epsilon \\
10\epsilon & 4-32\epsilon^2 & -2\epsilon & 2\epsilon^2 \\
1-\frac{37}{2}\epsilon^2 & -2\epsilon & 3-9\epsilon^2 & 2\epsilon \\
-2\epsilon & 2\epsilon^2 & 2\epsilon & \frac{16}{3}-\frac{80}{3}\epsilon^2
\end{array} \right]
$$
We are looking for a small eigenvalue of $\hat{L} x = \lambda x$, where $x = (x_0,x_1,x_2,x_{3})^T$.
Assuming $\lambda = \mathcal{O}(\epsilon^2)$ and expressing $\{x_1,x_3\}$ in terms of $\{x_0,x_2\}$ yield
$$
\left\{ \begin{array}{l}
x_1 = -\frac{\epsilon}{2} (5 x_0 - x_2) + \mathcal{O}(\epsilon^3),\\
x_3 = \frac{3 \epsilon}{8} (x_0 - x_2) + \mathcal{O}(\epsilon^3). \end{array} \right.
$$
Substituting these expressions into the eigenvalue problem $\hat{L} x = \lambda x$
and truncating it up to and including the order of $\mathcal{O}(\epsilon^2)$, we obtain
$$
\left\{ \begin{array}{l}
\left( \frac{1}{3} - \frac{125}{12} \epsilon^2 - \lambda \right) x_0 + \left( 1 - \frac{51}{4} \epsilon^2 \right) x_2 = 0, \\
\left( 1 - \frac{51}{4} \epsilon^2 \right) x_0 + \left( 3 -\frac{43}{4} \epsilon^2 - \lambda \right) x_2 = 0. \end{array} \right.
$$
The small eigenvalue of this reduced problem is given by the expansion
\begin{equation}
\label{eigen-10}
\lambda = -\frac{14}{5}\epsilon^2 + \mathcal{O}(\epsilon^4).
\end{equation}
Hence $n(L_+) = 1$ due to the shift of the zero eigenvalue of $L_+$ at $\epsilon$ to $\lambda < 0$ for $\epsilon \neq 0$.

For $\omega_6$, the Lyapunov--Schmidt decomposition of Theorem \ref{theorem-simple-N-1} yields
power expansion
$$
\omega = \frac{3}{35} + \frac{9}{70} \epsilon^2 + \mathcal{O}(\epsilon^4),
$$
with $A_1 = 1 - \epsilon^2 + \mathcal{O}(\epsilon^4)$, $A_6 = \epsilon$,
$A_{11} = -\frac{7}{8} \epsilon^2 + \mathcal{O}(\epsilon^4)$,
$A_n = \mathcal{O}(\epsilon^3)$ for $n \geq 16$, where $A_n = 0$ for every $n \neq 5 \ell + 1$, $\ell \in \mathbb{N}$
by Lemma \ref{lem-nonlinear-N-1}. Similarly to the proof of Lemma \ref{lemma2},
we compute the $3$-by-$3$ block of the operator $L_+$ at $n = 1$, $n = 6$, and $n = 11$
and truncate it up to and including $\mathcal{O}(\epsilon^2)$ terms. The corresponding
matrix block denoted by $\hat{L}_+$ is given by
$$
\hat{L}_+ = \left[ \begin{array}{ccc} 4 - 8 \epsilon^2 & 8 \epsilon & -7 \epsilon^2 \\
8 \epsilon & 14 \epsilon^2 & 4 \epsilon \\
-5 \epsilon^2 & 4 \epsilon & \frac{16}{7} + \frac{150}{7} \epsilon^2 \end{array} \right]
$$
Looking at the small eigenvalue of $\hat{L} x = \lambda x$, where $x = (x_1,x_6,x_{11})^T$,
we can normalize $x_6 = 1$ and obtain
$$
x_1 = -2\epsilon + \mathcal{O}(\epsilon^3), \quad
x_{11} = -\frac{7}{4} \epsilon + \mathcal{O}(\epsilon^3),
$$
and
\begin{equation}
\label{eigen-11}
\lambda = -9\epsilon^2 + \mathcal{O}(\epsilon^4).
\end{equation}
Hence $n(L_+) = 1$ due to the shift of the zero eigenvalue of $L_+$ at $\epsilon$ to $\lambda < 0$ for $\epsilon \neq 0$.
\end{proof}

Finally, we consider the two constraints related to the fixed values of $Q$ and $E$ by
using the same computational formulas (\ref{neg-index}) and (\ref{D-matrix}).
Let $\tilde{A}$ denote the stationary state of the stationary equation (\ref{sys}) continued with respect to
two parameters $(\lambda,\omega)$. By the scaling transformation (\ref{symmscale}),
if the stationary state is given by (\ref{stac}) with real $A$,
then the stationary state is continued with respect to parameter $c > 0$ as
$$
\tilde{\alpha}_n(t) = c A_n e^{-i c^2 \lambda t + i n c^2 \omega t},
$$
hence $\tilde{A} = c A$, $\tilde{\lambda} = c^2 \lambda$, and $\tilde{\omega} = c^2 \omega$.
Substituting these relations into $\mathcal{\tilde{Q}}(\tilde{\lambda},\tilde{\omega}) = Q(\tilde{A})$ and
$\mathcal{\tilde{E}}(\tilde{\lambda},\tilde{\omega}) = E(\tilde{A})$ for $\lambda - \omega = 1$ yields
$$
\mathcal{\tilde{Q}}(\tilde{\lambda},\tilde{\omega}) = c^2 \mathcal{Q}_0(\omega) =
(\tilde{\lambda} - \tilde{\omega}) \mathcal{Q}_0(\tilde{\omega} (\tilde{\lambda} - \tilde{\omega})^{-1})
$$
and
$$
\mathcal{\tilde{E}}(\tilde{\lambda},\tilde{\omega}) = c^2 \mathcal{E}_0(\omega) = (\tilde{\lambda} - \tilde{\omega})
\mathcal{E}_0(\tilde{\omega} (\tilde{\lambda} - \tilde{\omega})^{-1}),
$$
where $\mathcal{Q}_0(\omega) = \mathcal{Q}(1+\omega,\omega)$ and $\mathcal{E}_0(\omega) = \mathcal{E}(1+\omega,\omega)$.
Substituting these representations into (\ref{D-matrix}), evaluating derivatives, and setting $c = 1$ yield
the computational formula
\begin{equation*}
D = \left[ \begin{array}{cc} \mathcal{Q}_0(\omega) - \omega \mathcal{Q}'_0(\omega) & - \mathcal{Q}_0(\omega) + (1+\omega) \mathcal{Q}'_0(\omega) \\
\mathcal{Q}_0(\omega) - \mathcal{E}_0(\omega) - \omega \left[ \mathcal{Q}'_0(\omega) - \mathcal{E}'_0(\omega) \right] &
-\mathcal{Q}_0(\omega) + \mathcal{E}_0(\omega) + (1+\omega) \left[ \mathcal{Q}'_0(\omega) - \mathcal{E}'_0(\omega)\right] \end{array} \right],
\end{equation*}
which can be used to compute $D$ for the normalized stationary state $A$ with $\lambda - \omega = 1$.
The following lemma confirms that
the two branches in Lemma \ref{lemma2-N-1} are indeed local minimizers of $H$ subject to fixed $Q$ and $E$.

\begin{lemma}
Consider the two bifurcating branches in Theorem \ref{theorem-simple-N-1} for
$\omega_2$ and $\omega_6$. For every small $\epsilon \neq 0$, the two branches
are local minimizers of $H$ subject to fixed $Q$ and $E$
with $n_c(L_+) = 0$ and $n(L_-) = 0$.
\label{lemma3-N-1}
\end{lemma}

\begin{proof}
For $\omega_2$, we use power expansions in Lemma \ref{lemma2-N-1} and compute
\begin{eqnarray*}
\left\{ \begin{array}{l}
\mathcal{Q}_0(\omega) = A_0^2 + 2 A_1^2 + 3 A_2^2 + \dots = 2 - 4 \epsilon^2 + \mathcal{O}(\epsilon^4), \\
\mathcal{E}_0(\omega) = A_0^2 + 4 A_1^2 + 9 A_2^2 + \dots = 4 - 14 \epsilon^2 + \mathcal{O}(\epsilon^4),
\end{array} \right.
\end{eqnarray*}
so that
\begin{eqnarray*}
D = \left[ \begin{array}{cc} 6/7 + \mathcal{O}(\epsilon^2) & 6/7 + \mathcal{O}(\epsilon^2) \\
6/7 + \mathcal{O}(\epsilon^2) & -36/7 + \mathcal{O}(\epsilon^2) \end{array} \right].
\end{eqnarray*}
Note that the expression for $D$ agrees with the exact computations in Remark \ref{remark-pair}.
Therefore, $D$ has one positive and one negative eigenvalue, so that $n_c(L_+) = n(L_+) - 1 = 0$
by (\ref{neg-index}).

For $\omega_6$, we use power expansions in Lemma \ref{lemma2-N-1} and compute
\begin{eqnarray*}
\left\{ \begin{array}{l}
\mathcal{Q}_0(\omega) = 2 A_1^2 + 7 A_6^2 + \dots = 2 + 3 \epsilon^2 + \mathcal{O}(\epsilon^4), \\
\mathcal{E}_0(\omega) = 4 A_1^2 + 49 A_6^2 + \dots = 4 + 41 \epsilon^2 + \mathcal{O}(\epsilon^4),
\end{array} \right.
\end{eqnarray*}
so that
\begin{eqnarray*}
D = \left[ \begin{array}{cc} \mathcal{O}(\epsilon^2) & 70/3 + \mathcal{O}(\epsilon^2) \\
70/3 + \mathcal{O}(\epsilon^2) & -2726/81 + \mathcal{O}(\epsilon^2) \end{array} \right]
\end{eqnarray*}
has again one positive and one negative eigenvalue, hence $n_c(L_+) = n(L_+) - 1 = 0$
by (\ref{neg-index}).
\end{proof}

\section{Numerical approximations}

We confirm numerically that the stationary states (\ref{min-intro}) with (\ref{omega-lambda-intro})
have the same variational characterization in the entire existence interval for $p$ in $(0,2-\sqrt{3})$.
Therefore, they remain constrained minimizers of $H$ for fixed $Q$ and $E$.

Figure \ref{fig-1} shows the smallest eigenvalues of $L_{\pm}$ computed at the upper branch of solution
(\ref{min-intro}). In agreement with item (iii) in Lemma \ref{lemma2} for small $p$,
we have $n(L_+) = 1$ and $n(L_-) = 0$ for every $p \in (0,2-\sqrt{3})$.
The small parameter $p$ is related to the small parameter $\delta$ in Lemma \ref{lemma2}
by $p = \delta + \mathcal{O}(\delta^5)$, see Remark \ref{remark-state}. The dashed line shows
the asymptotic dependencies (\ref{eigen-7}), (\ref{eigen-8}), and (\ref{eigen-9})
for the small eigenvalues of $L_+$ and $L_-$.

\begin{figure}[h]
\centering
\includegraphics[width=0.45\textwidth]{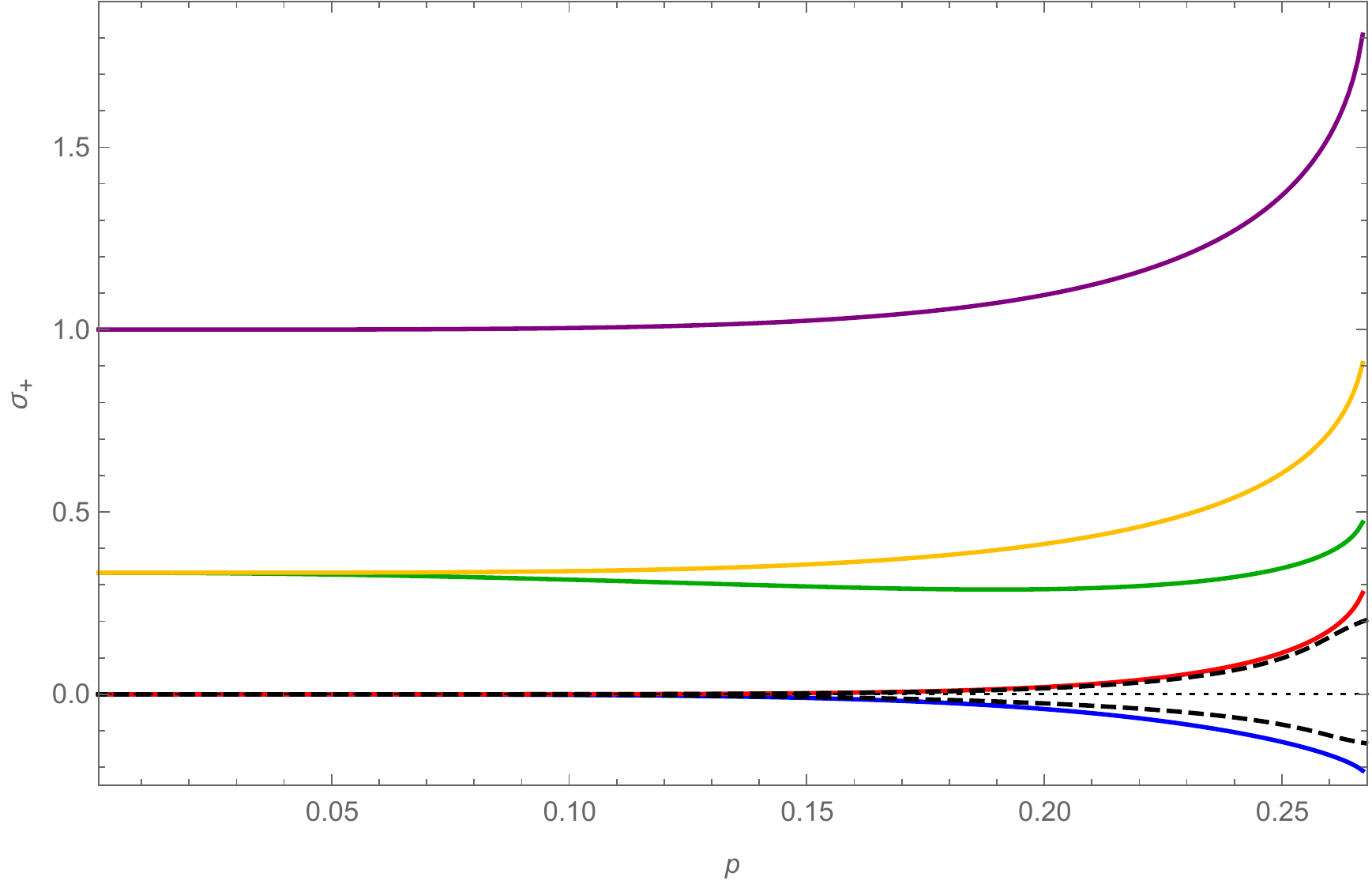}
\includegraphics[width=0.45\textwidth]{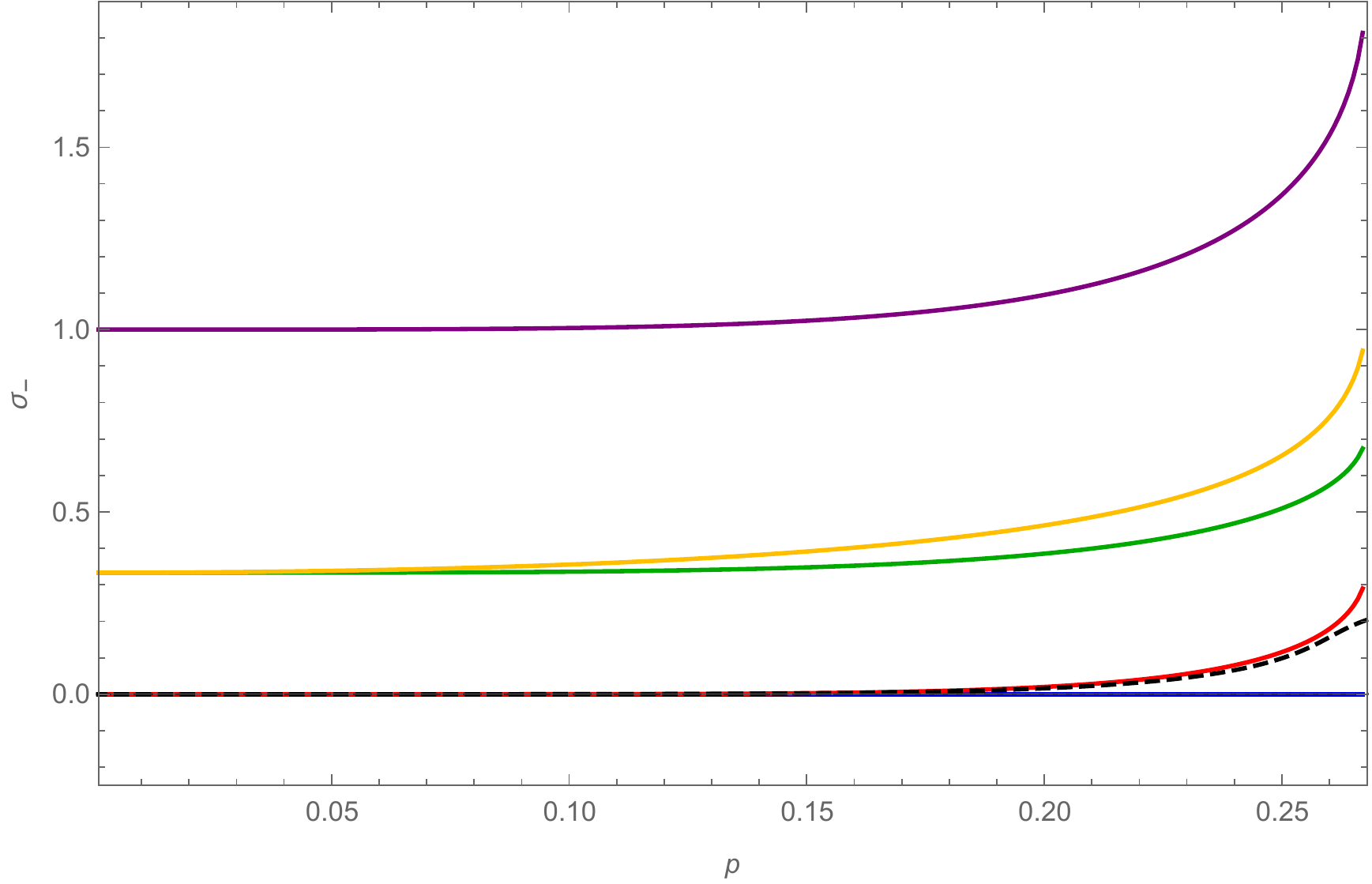}
\caption{The smallest eigenvalues of $L_+$ (left) and $L_-$ (right)
for the upper branch of the stationary state (\ref{min-intro}) with normalization $\lambda=1$. }
\label{fig-1}
\end{figure}

Figure \ref{fig-2} shows the smallest eigenvalues of $L_{\pm}$ computed at the lower branch of solution
(\ref{min-intro}). In agreement with Lemma \ref{lemma2-N-1} for small $p$,
we have $n(L_+) = 1$ and $n(L_-) = 0$ in the entire region of existence of the stationary state.
The small parameter $p$ is related to the small parameter $\epsilon$ in Lemma \ref{lemma2-N-1}
by $p = -\epsilon/2 + \mathcal{O}(\epsilon^3)$, see Remark \ref{remark-state-minus}.
The dashed line shows the asymptotic dependence (\ref{eigen-10}) for the small eigenvalue of $L_+$.

\begin{figure}[h]
\centering
\includegraphics[width=0.45\textwidth]{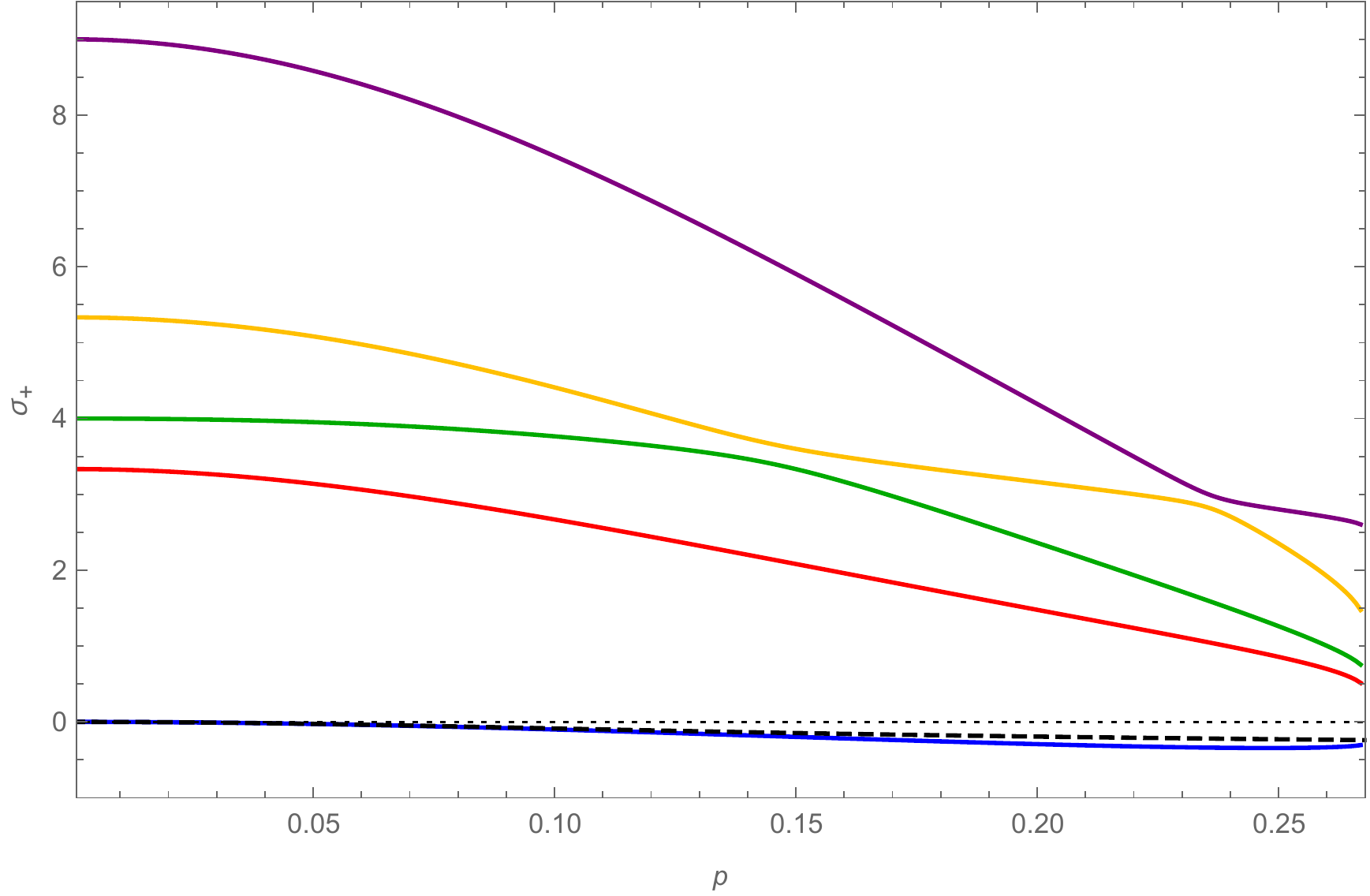}
\includegraphics[width=0.45\textwidth]{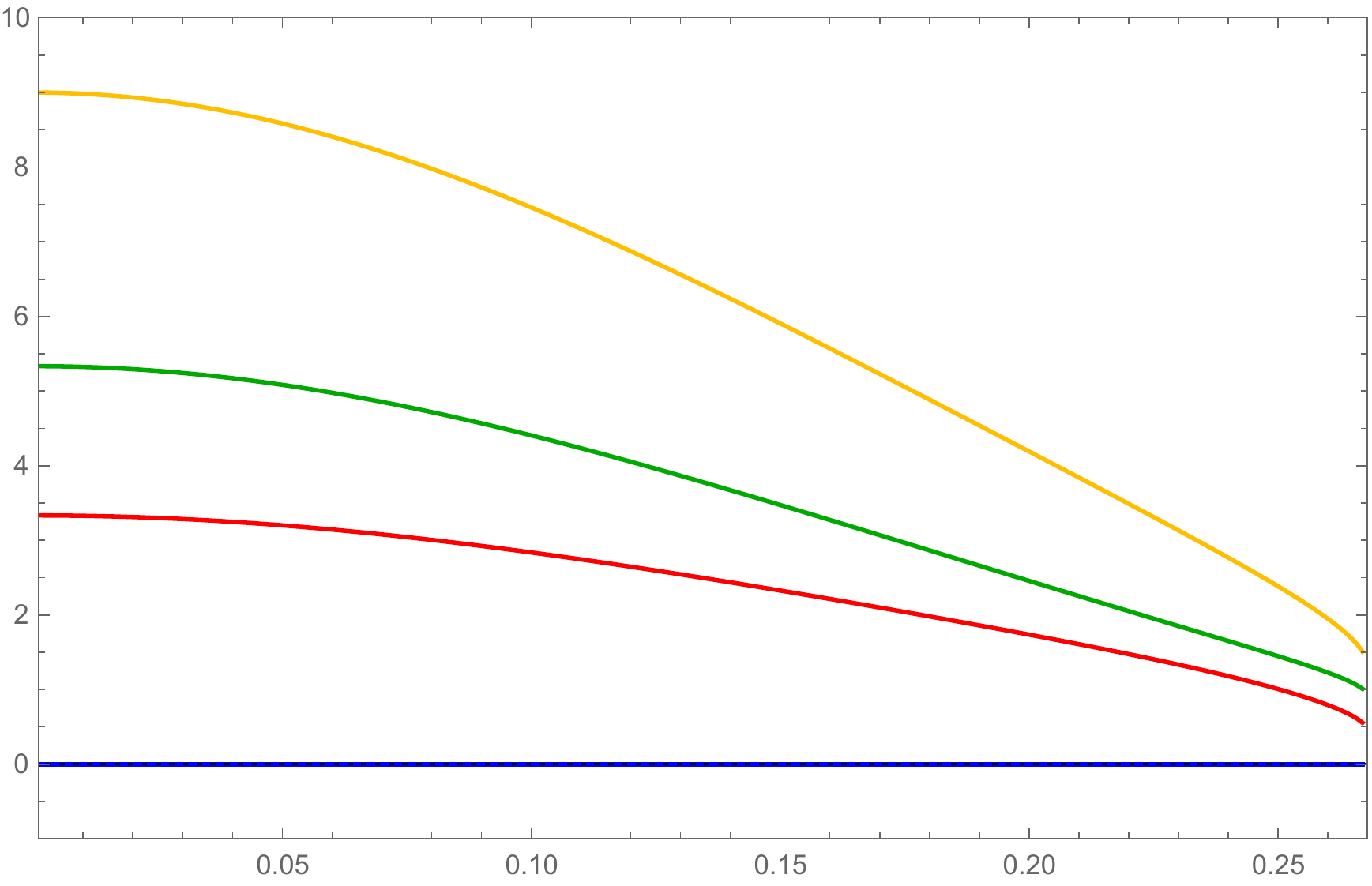}
\caption{The smallest eigenvalues of $L_+$ (left) and $L_-$ (right)
for the lower branch of the stationary state (\ref{min-intro}) with normalization $\lambda-\omega=1$.}
\label{fig-2}
\end{figure}

Figure \ref{fig-3} shows the smallest eigenvalues of $L_{\pm}$ computed at the stationary state
bifurcating from the second eigenmode at $\omega_6 = 3/35$. We use here parameter $\epsilon$
for continuation of the stationary state as in Lemma \ref{lemma2-N-1}.
In agreement with Lemma \ref{lemma2-N-1}, we have $n(L_+) = 1$ and $n(L_-) = 0$ for small $\epsilon$.
However, this result does not hold for larger values of $\epsilon$ far from the bifurcation point
because additional eigenvalues of $L_+$ and $L_-$ become negative eigenvalue for $\epsilon \approx 0.04$.
Therefore, the stationary state becomes a saddle point of $H$ for fixed $Q$ and $E$ when $\epsilon \gtrsim 0.04$.
The dashed line shows the asymptotic dependence (\ref{eigen-11}) for the small eigenvalue of $L_+$.

\begin{figure}[h]
\centering
\includegraphics[width=0.45\textwidth]{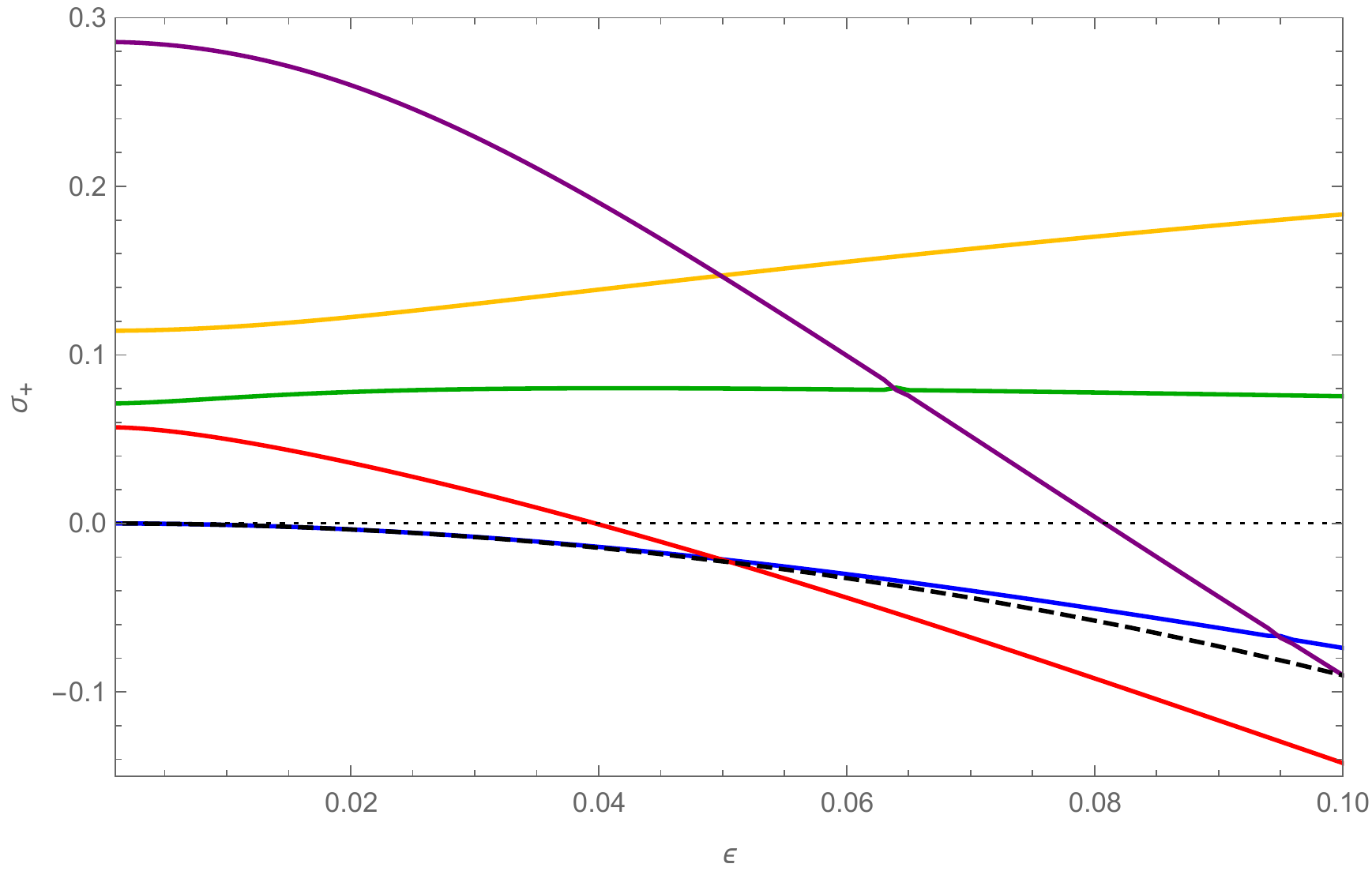}
\includegraphics[width=0.45\textwidth]{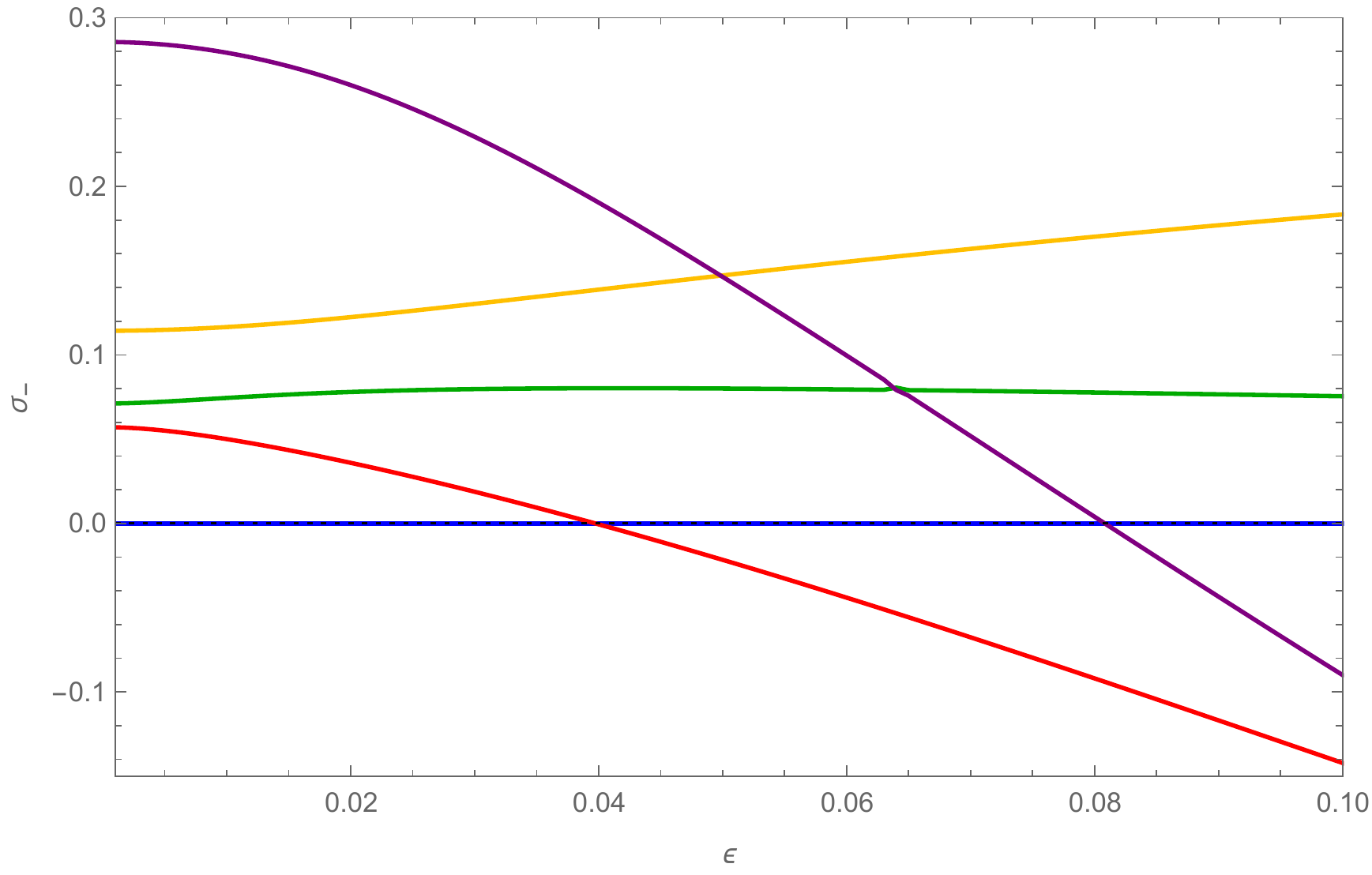}
\caption{The smallest eigenvalues of $L_+$ (left) and $L_-$ (right)
for the branch bifurcating from the second eigenmode at $\omega_6 = 3/35$ with normalization $\lambda-\omega=1$.}
\label{fig-3}
\end{figure}

\begin{rem}
The presence of zero eigenvalue in the spectrum of $L_+$ at $\epsilon \approx 0.04$ on Figure \ref{fig-3} singles
out new bifurcation of the stationary state along the branch. We have checked that the numerical results are stable
with respect to truncation. In the present time, it is not clear how to identify new solution branches which may
branch off at one or both sides of the bifurcation point.
\end{rem}

\end{document}